\documentclass[aos,preprint]{imsart}

\RequirePackage{amsthm,amsmath,amsfonts,amssymb,amsthm}
\RequirePackage[authoryear]{natbib} 
\RequirePackage[colorlinks,citecolor=blue,urlcolor=blue]{hyperref}
\RequirePackage{graphicx}
\usepackage{ifoddpage}

\startlocaldefs

\PassOptionsToPackage{unicode}{hyperref}
\PassOptionsToPackage{naturalnames}{hyperref}

\usepackage{amsmath,amssymb,amsfonts, mathrsfs, latexsym,natbib}
\usepackage{bbm}
\usepackage{makecell}
\def\nn{\mathbb{N}}

\usepackage[dvipsnames]{xcolor}
\newcommand*{\myfont}{}

\def\rr{\mathbb{R}}

\def\zz{\mathbb{Z}}

\def\nn{\mathbb{N}}

\def\xx{\mathbb{X}}

\def\argmin{\text{argmin}}

\def\1{\mathbbm{1}}

\allowdisplaybreaks
\usepackage{capt-of}

\renewcommand{\leq}{\leqslant} 
\renewcommand{\geq}{\geqslant}

\newcommand{\eps}{\varepsilon}

\def\qed{ \hfill $\blacksquare$}  

\newtheorem{theorem}{Theorem}[section]
\newtheorem{lemma}[theorem]{Lemma}
\theoremstyle{remark}

\endlocaldefs

\begin{document}

\begin{frontmatter}
\title{Introduction to Generalized Fiducial Inference}

\begin{aug}
\author[A]{\fnms{Alexander} \snm{Murph}\ead[label=e1]{acmurph@live.unc.edu}},
\author[A]{\fnms{Jan} \snm{Hannig}\ead[label=e2]{jan.hannig@unc.edu}}
\and
\author[B]{\fnms{Jonathan} \snm{P Williams}\ead[label=e3]{jwilli27@ncsu.edu}}
\address[A]{Department of Statistics \& Operations Research,
University of NC at Chapel Hill,
\printead{e1,e2}}

\address[B]{Department of Statistics,
NC State University, Raleigh, NC,
\printead{e3}}
\end{aug}

\begin{abstract}
Fiducial inference was introduced in the first half of the 20th century by 
\cite{fisher1935} as a means to get a posterior-like distribution for a parameter without having to arbitrarily define a prior.  While the method originally fell out of favor due to non-exactness issues in multivariate cases, the method has garnered renewed interest in the last decade.  This is partly due to the development of generalized fiducial inference, which is a fiducial perspective on generalized confidence intervals: a method used to find approximate confidence distributions.  In this chapter, we illuminate the usefulness of the fiducial philosophy, introduce the definition of a generalized fiducial distribution, and apply it to interesting, non-trivial inferential examples.
\end{abstract}

%

\end{frontmatter}

\section{Introduction}

\label{sec:1}

Fiducial inference attempts to find a middle ground between the frequentist and Bayesian perspectives.
The fiducial argument allows one to fit a posterior-like distribution on a target parameter $\theta$ in a way that is entirely data-driven and does not rely on a sometimes arbitrary prior selection.  This argument is based on inverting a \textit{data generating algorithm (DGA)} that associates data $\textbf{Y}$ to the parameters $\theta$ and a random component $U$ with a known distribution $F_0$, e.g. a vector of independent and identically distributed (iid) standard uniforms or standard Gaussians. This is often expressed as the relation $\textbf{Y} = A(U, \theta)$. Since the DGA is a function of a random variable with a known distribution, it immediately determines the likelihood $f(\mathbf{y}|\theta)$.  By solving the DGA for $\theta$, we get the distribution of our parameters, called the {\em generalized fiducial distribution} (GFD), that is entirely data-driven and does not require the use of Bayes' Theorem.  

While the strengths and limitations of {\em generalized fiducial inference} (GFI) continue to be explored, its usefulness has already been illustrated in numerous practical applications.  Recent work has applied the fiducial ideas to bio-equivalence \citep{mcnally2003,hannig2006a}, metrology problems \citep{hannig2003,wang2006a,wang2006b,hannig2007,wang2012b}, inter-laboratory experiments and international key comparison experiments \citep{iyer2004}.  It has also been used to tackle statistical problems at the forefront of modern topics in statistical research, such as in wavelet regression \citep{hannig2009} and extreme value estimation \citep{wandler2012b}, and has recently led to a creative new perspective on linear model selection \citep{williams2019} and vector autoregressive graph selection \citep{williams2019b}.

{\myfont As a brief motivating example of the fiducial argument, we define the DGA of a single observation from a normal distribution with known mean, ${Y} = \mu + U$, where $U \sim \mathcal{N}(0, 1)$ is our random quantity with known distribution, assumed to be independent of $\mu$.  We ascribe a distribution on the parameter $\mu$ by inverting the DGA: $\mu = {Y} - U \sim \mathcal{N}(Y, 1)$.  While inverting the DGA for this simple normal example is algebraically simple, inverting a general DGA  could be non-trivial.}

{\myfont Heuristically, a smooth DGA $A(U,\theta)$ taken as a function of $\theta$ behaves locally like a linear function near our observed data value $\mathbf{y}$.  Therefore, for each realization of the random quantity $U^\star$, there is a well-defined point $\theta^\star$ so that $A(U^\star,\theta^\star)$ is closest to $\mathbf{y}$.  
GFD is calculated  as the distribution of $\theta^\star$ using the implicit function theorem from the distribution of $U^\star$ conditional on the event that $\{A(U^\star,\theta^\star)\approx \mathbf{y}\}$.}

Formally, we define the GFD as a limit in the following way.  Defining the pseudo-inverse of the DGA using the optimization problem
\begin{equation}\label{eq:FIDopt}
 Q_\textbf{y}(u)=\argmin_{\theta^{\star}} {\| \textbf{y} - A(u, \theta^{\star} ) \|}.
\end{equation}
Typically $\|\cdot\|$ is either $\ell_2$ or $\ell_\infty$ norm.  {This will serve to find the closest point $\theta$ described above.  While not technically an inverse, this quantity is always defined.}
Next, for each small $\epsilon>0$, define the random variable $\theta_\epsilon^\star=Q_\textbf{y}(U_\epsilon^{\star})$,
where $U_\epsilon^{\star}$ has  distribution $F_0$ {\bf truncated} to the set 
\begin{equation}\label{eq:truncate}
\mathcal M_{\mathbf{y},\epsilon} = \{ U_\epsilon^{\star} : \|\mathbf{y} - A(U_\epsilon^{\star}, \theta_\epsilon^\star ) \| = \| \mathbf{y} - A(U_\epsilon^{\star}, Q_\mathbf{y}(U_\epsilon^{\star}) ) \| \leq \epsilon \},
\end{equation}
{\myfont i.e., having the density $f_U(u^\star) I_{ M_{\mathbf{y},\epsilon}}(u^\star)/(\int_{M_{\mathbf{y},\epsilon}} f_U(u)\,du),$ where $f_U$ is the original density of $U$.}
Then assuming that the random variable $\theta_\epsilon^\star$ converges in distribution as $\epsilon\to 0$, the GFD is defined as the limiting distribution of $\lim_{\epsilon\to 0} \theta_\epsilon^\star$.

When the sampling distribution of $\mathbf{Y}$ is discrete, we can set $\epsilon=0$ and no limit is necessary.
When the sampling distribution of $\mathbf{Y}$ is continuous, {\myfont \cite{GenFid} use the implicit function theorem to calculate the limiting distribution of $\theta_\epsilon^\star$}, which leads to the following result.
\begin{theorem} 
	\label{Jacobian}
	Under mild conditions (see  Assumptions B.1-B.4 in \citealt{GenFid}), the limiting distribution above has density 
	\begin{equation}\label{eq:Jacobian}
r_{\mathbf{y}} (  { \theta } ) = \frac { f( \mathbf{y} |  { \theta } ) J ( \mathbf{y} ,  { \theta } ) } { \int  f \left( \mathbf{y}  |  { \theta } ^ { \prime } \right) J \left(\mathbf{y} ,  { \theta } ^ { \prime } \right) d  { \theta } ^ { \prime } },	    
	\end{equation}
	where 
	$J ( \mathbf{y} ,  { \theta } ) = D \left( { \nabla_{  { \theta } } } A\left. ( u ,  { \theta } ) \right| _ { u = A ^ { - 1 } ( \mathbf{y} ,  { \theta } ) } \right).$
	Here $\nabla_{  { \theta }} A(u,\theta)$ is the gradient matrix computed with respect to $\theta$, and $D$ is a determinant like operator that depends on the norm in \eqref{eq:FIDopt}, e.g., when we use $\ell_2$ norm $D( M ) = \left( \operatorname { det } M ^ { \prime } M \right)^{ \frac { 1 } { 2 } }$. 
\end{theorem}

The generalized fiducial approach communicates a simple algorithm: when possible, define a DGA that expresses the relationship between the data, the parameter, and a random quantity, then invert it.  The application of this idea in practice, however, can be nuanced and requires careful thought.  Our aim in this chapter is to provide a comprehensive overview of GFI by way of detailed examples.  {\myfont As we will illustrate below, in some examples inverting a DGA does not lead to a single point associated with a given $U$ draw from $F_0$. In this case we will find a fiducial distribution on the sets of parameters, cf. Demspter-Shafer theory \citep{dempster}}

In particular, we will treat two instances of multivariate normal data and a binomial distribution with unknown number of trials. We selected these examples both because they are of interest in their own right, but also because they will allow us to demonstrate how to implement a generalized fiducial solution using modern computational tools.
The computer codes are available at a GitHub repository [https://github.com/sirmurphalot/IntroductionGFI].  {\myfont For the binomial algorithms, this GitHub page also includes full pseudocodes of the implementations.}

{\myfont Throughout this paper, numerous approximate fiducial confidence intervals are reported.  Although we use the term \textit{confidence}, our approximate fiducial confidence intervals are closely related to the notion of the Bayesian \textit{credible} interval.  Each of these intervals involve defining a probability distribution on a target parameter and using this distribution to calculate a set whose probability mass matches a researcher's desired level of coverage.  Rather than call the fiducial intervals credible intervals, which are specifically a Bayesian construct, we refer to them as approximate confidence intervals.  

The caveat that these intervals are \textit{approximate} is necessary for numerous reasons.  For the continuous multivariate normal problems, it is necessary to perform Markov Chain Monte Carlo (MCMC) sampling, an estimation method, to circumvent the need to calculate the untenable marginal integral in the denominator of Equation \eqref{eq:Jacobian}.  MCMC sampling is also necessary to explore the sample space for the binomial problem for $n$ and $p$ unknown.  In addition, for both of the binomial problems, approximation is necessary to address the issue of an unbounded sample space on the $n$ parameter.  While all these sources of uncertainty merit the use of the term ``approximate," the primary reason we refer to these intervals in this way is because the coverages need not be exact.  This ``close but not exact'' coverage is analogous to the caveats one imposes on frequentist confidence intervals based on asymptotic normality.  We address these multiple sources of approximation in this paper using simulation studies to assess whether the performance of our computational strategies (the \textit{empirical coverage}) achieve our theoretical expectations (the \textit{nominal coverage}).  The results of these simulations show that, regardless of these approximations, the fiducial method is justified for practical use.
}

\section{Multivariate Normal Distribution}
\subsection{Multivariate Normal Data; \texorpdfstring{$\mu$}{} unknown, \texorpdfstring{$\Sigma$}{} unknown}

The estimation of covariance matrices is a fundamental problem in many multivariate methods.  Examples include discriminant data analysis, longitudinal data analysis, time series analysis, and spatial data analysis, just to name a few.  However, only recently it was pointed out that in the Bayesian context the most commonly used conjugate inverse Wishart prior {\myfont may not be the best choice for estimation}  \citep{BergerSunSong2018a,  BergerSunSong2018b, YangBerger1994}. {\myfont In particular, the inverse Wishart posterior has the effect of forcing the eigenvalues of the covariance matrix apart.  This ``systemic distortion" of the eigenstructure of the covariance matrix has negative effects on estimation, and has been studied extensively.  For further details on this problem, and a survey of proposed solutions, we suggest the papers by \cite{perron1993} and \cite{YangBerger1994}.}

Let $\mathbf Y_1, \mathbf Y_2, \dots, \mathbf Y_m$ be iid  $\mathcal{N}_d( \mu, \Sigma)$, where $\mu$ is a $d$-dimensional vector and $\Sigma$ is a $d \times d$ covariance matrix.  Our aim is to perform inference on the covariance matrix $\Sigma$ with the unknown mean parameter $\mu$.
To date there have been proposed two basic approaches to define the GFD for this model. They both start with the DGA,
\[
 \mathbf{Y}_i=\mu + B \mathbf{U}_i, \quad i=1,\ldots,n,
\]
where $\mathbf U_i$ are iid standard Gaussian vectors of dimension $d$, but differ in the structure of the matrix $B$. \cite{wandler2011} use $B$ that is a lower triangular matrix. This leads to GFD that depends on the arbitrary order of the coordinates.  \cite{ShiHannigLaiLee2017} propose using an arbitrary $d\times d$ matrix $B$, which removes the dependence on the coordinate order but is overparametrized.  {\myfont In the second case, the resulting GFD for the covariance matrix $\Sigma=BB^\top$ belongs to the Wishart family.   Several other reasonable choices for $B$ lead to a GFD belonging to the Wishart family:  $B = \Sigma^{1/2} = Z\Lambda Z^\top$,  where $Z$ is an orthogonal matrix, $\Lambda$ is a diagonal matrix with positive entries on the diagonal, or the lower triangular $B$ with a DGA  for the sufficient statistic, i.e., the sample mean and sample covariance modeled by $\bar{\mathbf{Y}} = \mu + n^{-1} BU_1,\ S^2 = n^{-1} B U_2 B^\top$ respectively, where $U_1$ is $d$-variate standard normal and $U_2$ has Wishart distribution $W_d(I,n-1)$.}

We propose an alternative DGA that does not lead to the Wishart distribution and is not overparameterized. In particular, consider $B= Z\Lambda$, where $Z$ is an orthogonal matrix, $\Lambda$ is a diagonal matrix with positive entries on the diagonal. Consequently, the covariance matrix is $\Sigma=Z\Lambda^2Z^\top$. To be able to compute the GFD using  \eqref{eq:Jacobian} we will need to reparametrize $Z$ using the {\it Cayley transformation} (see Theorem \ref{Cayley_transform}). In particular we will use the following two facts. For the proof of the first fact see, for instance, \cite{howard1996}.

\begin{theorem}[Cayley transform]\label{Cayley_transform}
Every real orthogonal matrix $Z$ that does not have -1 as a characteristic root can be expressed as
\begin{equation}\label{eq:ourU} 
Z = (I_d-A)(I_d+A)^{-1}=(I_d+A)^{-1}(I_d-A) \end{equation}
by a suitable choice of a real skew-symmetric matrix, i.e., $A^\top=-A$.
\end{theorem}    

\begin{theorem}[\citealt{odorney2014}]
    For any orthogonal matrix $Z$ there must exist a \textit{signature matrix} $D$ such that $ZD$ does not have -1 as a characteristic root and all elements of corresponding Cayley transform $a_{i,j}$ are such that $|a_{i,j}|\leq 1$ for $1\leq i\leq j \leq d$.  
\end{theorem}

Recall that $D$ is a diagonal matrix with entries $\pm 1$. Consequently, $ZD\Lambda^2 DZ^\top=Z\Lambda^2 Z^\top$ which is invariant to the choice of $D$.
This leads us to propose the following data generating algorithm
\begin{equation}\label{eq:MVNDGA}
 \mathbf Y_i = \mu + (I_d-A)(I_d+A)^{-1}\Lambda \mathbf U_i, \quad i=1,\ldots,n, 
\end{equation}
where $A$ is a skew-symmetric matrix with all entries $|a_{ij}|\leq 1$ and $\Lambda$ is a diagonal matrix with positive entries $\lambda_i>0$.

We show in Appendix~\ref{a:MVN} that the 
Jacobian is $J(\mathbf y,\theta)=J^*(\mathbf y,A) \prod_{i=1}^d\lambda_i^{-1} $, where $J^*(\mathbf y,A)$ does not depend on $\mu$ or $\lambda$. 
The form of this Jacobian allows us to simplify the generalized fiducial density from Equation \eqref{eq:Jacobian}. {\myfont In particular, calculations in Appendix~\ref{a:MVN} show that the marginal GFD of \textit{the unique entries of} $A$ is\begin{equation}\label{eq:MVNmargA}
r_{\mathbf y}(\operatorname{veck}(A) ) \propto J^*(\mathbf y,A) \prod_{i=1}^d( Z^T   nS^2 \, Z)_{ii}^{\frac{-(n-1)}{2}},
\end{equation}
where $S^2=\frac{1}{n}\sum_{i=1}^n(\mathbf Y_i-\mathbf{\bar Y})(\mathbf Y_i-\mathbf{\bar Y})^\top$ and $\mathbf{\bar Y}=n^{-1}\sum_{i=1}^n \mathbf Y_i$.  The $\operatorname{veck}(A)$ operation vectorizes the strictly lower triangular elements of the skew-symmetric matrix $A$, as discussed in \cite{henderson1979}.  The matrix $Z$ is defined as a function of $\operatorname{veck}(A)$ in Equation \eqref{eq:ourU}.  The conditional GFD of the diagonal entries of $\Lambda^{-2}$ given the $\operatorname{veck}(A)$ follow independent gamma distributions, $\lambda_i^{-2}\sim$ Gamma$\left(\frac{n-1}2,\frac{ (Z^\top nS^2 \, Z)_{ii}}{2}\right),\ i=1,\ldots, d$.  The conditional GFD of $\mu$ given $\Lambda$ and $A$ is multivariate $\mathcal{N}\left(\mathbf{\bar Y},n^{-1}{Z\Lambda^2 Z^\top}\right)$.}  The simple form of this distribution allows us to implement sampling from the GFD using STAN \citep{STAN}.  
 
{\myfont Since $A$ is skew-symmetric, its diagonal terms must be zero, and the matrix is determined by the entries of the strict lower triangle of $A$. The number of free parameters $(A,\Lambda,\mu)$ in \eqref{eq:MVNDGA} is therefore ${d(d +3)}/{2}$, which matches the number of free equations in the minimal sufficient statistic $(\mathbf{\bar Y},S^2)$, where there are $d$ free equations in the vector $\mathbf{\bar Y}$ and $(d+1)d/2$ free equations from the lower triangle of $S^2$.} Consequently, our DGA is not overparameterized.  We do not impose an order restriction on the diagonal entries of $\Lambda$.  This makes the calculation of the entries $\lambda_i$ easier, but introduces non-uniqueness.  To address this issue, we run parallel chains of our MCMC algorithm each starting with their own random ordering of the singular values, obtained by using \textit{Principle Component Analysis} on the sample covariance matrix.

The STAN software \citep{STAN} allows us to quickly and easily draw samples from the GFD, which we can then use to obtain approximate fiducial confidence intervals for the true covariance matrix.  To build these fiducial confidence intervals and evaluate our method, we consider a number of distance metrics and parameters of interest.  For the distance metrics, fiducial confidence intervals are developed by defining a ball around the mean of the GFD that covers (1 - $\alpha$)\% of the GFD.  {\myfont This is done by calculating the distance from the average point for every value sampled from the GFD and defining the ball using a distance cutoff such that $(1-\alpha)\%$ of the sampled values have a distance from the average value that is less the chosen cutoff.} For parameters of interest, fiducial confidence intervals are developed much like they would be on the real line.  That is, after mapping every matrix from the GFD to the real line, we simply take the center $(1-\alpha)\%$ of the GFD as our fiducial confidence interval.

In particular, we construct the fiducial confidence intervals for the true covariance matrix using the following two distance metrics, two norms, and one parameter of interest. The norms are used as both distance metrics (for matrices $M,N: ||M - N||$) and as parameters of interest (for matrix $M: ||M||$).   Let $M,N$ be two positive-definite, symmetric matrices.  Then,
\begin{enumerate}
    \item $\textsc{FM\_Distance} \text{ \citep{Forstner2003}}: dist(M,N) = \sqrt{\sum_{i=1}^d \ln^2\lambda_i(M,N)}$, 
    where $\lambda_i(M,N)$ are the eigenvalues from the equation $\det(\lambda M - N) = 0$;
    \item $\textsc{Stein\_Loss}\text{ \citep{Konno1995}}:  dist(M,N) = \text{tr}(M^{-1}N) - \log(\det(M^{-1}N)) - d$;
    \item $\textsc{Spectral\_Norm}\text{ \citep{horn2012}}: norm(M) = \lambda_{(d)}$, the maximum value in the $\Lambda$ matrix;
    \item $\textsc{Frobenius\_Norm}\text{ \citep{horn2012}}: norm(M) = \sqrt{\text{tr}(M M^{H})}$, where $M^{H}$ is the conjugate transpose of $H$;
    \item $\textsc{LogDet} \text{ \citep{Fazel2003}}: huer(M) = \log(\det(M))$.
\end{enumerate}

Our simulation study consists of 1000 iterations on a $4 \times 4$ multivariate normal $\mathcal{N}_4(\mu, \Sigma)$ where
\[ \mu = \begin{pmatrix} 1 \\ 2 \\ 3 \\ 1 \end{pmatrix}, ~~\Sigma = \begin{pmatrix}4&1&0&0\\ 1&1&0&1\\ 0&0&9&1\\ 0&1&1&4 \end{pmatrix}. \]  
At each iteration, we simulate a dataset of $n=100$ observations from the given multivariate normal distribution. Then, we run 20 MCMC chains and use the simulated values from the GFD of the covariance matrix to construct approximate fiducial confidence intervals around the true covariance matrix $\Sigma$.  

\begin{figure}[ht]
  \checkoddpage
  \edef\side{\ifoddpage l\else r\fi}%
  \centering
  \makebox[\textwidth][c]{%
    \begin{minipage}[t]{0.59\textwidth}
      \centering
      \includegraphics[width=\linewidth]{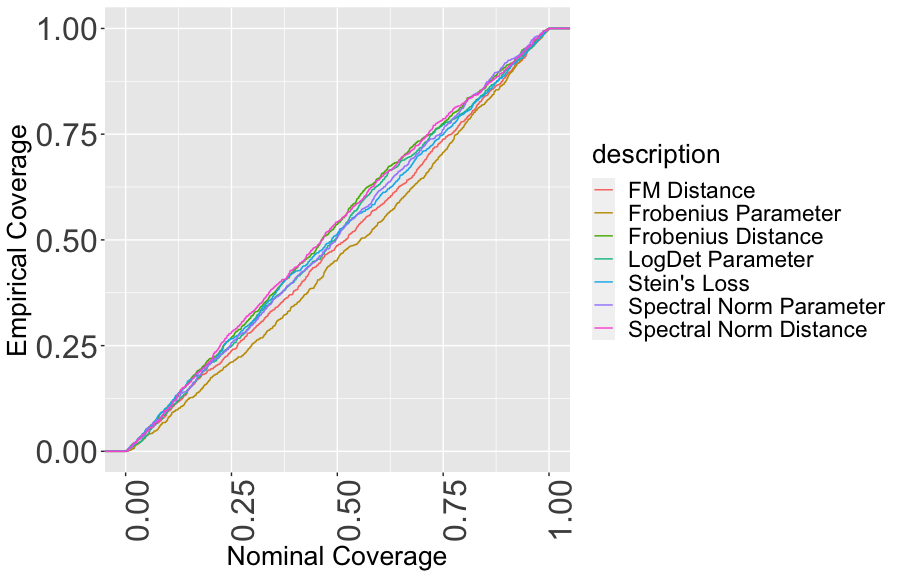}
  \captionof{figure}{The coverage of fiducial confidence intervals for the covariance matrix using different notions of distance.  This figure is created by setting different levels of coverage in our GFD fiducial confidence intervals and checking the empirical coverage of these intervals over the whole simulation.  A perfect diagonal line equates to a match between our observed (empirical) coverage and our theoretical (nominal) coverage.  {\myfont These lines show that our method approximately achieves the desired coverage.}}
  \label{fig:MVNQQs}
    \end{minipage}%
    \hspace*{1cm}
    \begin{minipage}[t]{0.59\textwidth}
      \centering
      \includegraphics[width=\linewidth]{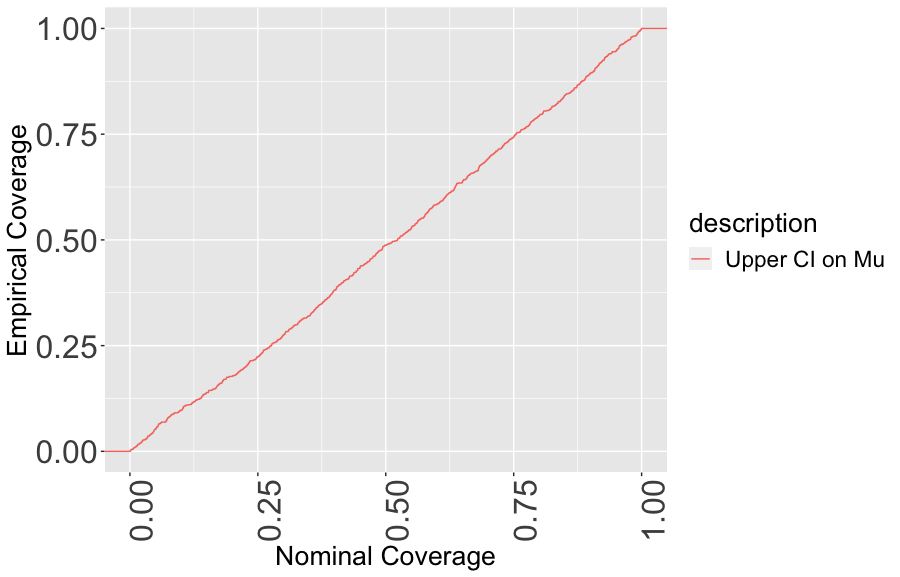}
  \captionof{figure}{The coverage of fiducial confidence intervals for the mean vector $\mu$ using standard euclidean distance.  This figure was made in the same way as figure \ref{fig:MVNQQs}.  {\myfont This line shows that our method approximately achieves the desired coverage.}}
  \label{fig:MVNMu}
    \end{minipage}%
  }%
\end{figure}

Figure \ref{fig:MVNQQs} shows the performance of these fiducial confidence intervals for different values of $\alpha$.  The nominal coverage corresponds to the value of $\alpha$ we set and the empirical coverage is the proportion of times (1 - $\alpha$) balls made this way contained the true covariance matrix.  As we can see, all of our metrics resulted in a line across the diagonal, communicating that nominal and empirical coverage match.  Table \ref{tab:MVN95Ints} shows the specific coverage of these fiducial confidence intervals for $\alpha = 0.05$.  As we can see, the empirical coverage probabilities of each of these intervals is approximately 95\%, as desired.

\begin{table}
\centering
\begin{minipage}{.9\textwidth}
  \centering
  \begin{center}
 \begin{tabular}{||c | c ||} 
 \hline
  Distance Metric/Norm & {\myfont Empirical Coverage of the 95\% Fiducial Confidence Interval} \\
 \hline
 $\textsc{FM\_Distance}$ & 0.9505 \\
 $\textsc{Loss1\_Distance}$ & 0.9463 \\
 $\textsc{LogDet}$ Parameter & 0.9589 \\
 $\textsc{Spectral\_Norm}$ Parameter & 0.9568 \\
 $\textsc{Frobenius\_Norm}$ Parameter & 0.9547\\
 $\textsc{Spectral\_Norm}$ Distance & 0.9540 \\
 $\textsc{Frobenius\_Norm}$ Distance & 0.9580 \\
 \hline
\end{tabular}
\end{center}
  \captionof{table}{The empirical coverage of fiducial confidence intervals for the covariance matrix using different notions of difference.  {\myfont For the distance metrics, fiducial confidence intervals are developed by defining a ball around the mean of the GFD that covers 95\% of the GFD.  For parameters of interest, fiducial confidence intervals are developed much like they would be on the real line.  That is, after mapping every matrix from the GFD to the real line, we simply take the center 95\% of the GFD as our fiducial confidence interval.  This table shows how our empirical coverage is approximately equal to our expected nominal coverage of 95\%. }}
  \label{tab:MVN95Ints} 
\end{minipage}
\centering
{\small \\} 
\end{table}

The simultaneous fiducial confidence interval for $\mu$ is constructed the same way as $\Sigma$ using the standard Euclidean distance $dist(u,v) = ||u - v||$.  Draws from the GFD for $\mu$ were calculated at the same time as the draws of $\Sigma$ in our simulation study.  Figure \ref{fig:MVNMu} shows that nominal coverage matched the empirical coverage, validating our method in the context of this inference problem.

\subsection{Generalization of one way random effects model}
Continuing to examine the multivariate normal problem, we consider an alternative, more restrictive parameterization, the standard unbalanced one-way random effect model $Y_{i,j} = \mu + \eta_{i} + \epsilon_{i,j}$, for $\eta_i\sim \mathcal{N}(0, \sigma_a^2),~\epsilon_{i,j} \sim \mathcal{N}(0, \sigma_e^2)$, where $i = 1, \dots, n$ and $j = 1,\dots, d$, and the $\eta_i$ and $\epsilon_{i,j}$ terms are independent.  The first fiducial solution to inference on this model was given by \cite{e2008} and used a {\myfont tailor-made} solution to this problem. Let us consider the more general DGA
\begin{equation}\label{eq:MMDGA}
\textbf{Y} = \xx \beta + A \mathbf{U},
\end{equation}
where $\beta$ is an unknown vector of fixed effects, $\mathbf{U}$ is standard normal random vector, $\xx$ is a fixed effect design matrix and $A = \Sigma^{1/2}$ such that $\Sigma = \sigma_\alpha^2 S_\alpha+ \sigma_e^2 I$ and $S_\alpha$ is a matrix of ones and zeros corresponding to the group sizes $n_1, n_2, \dots, n_m$.  Thus, our inferential problem is simplified to providing GFD for $\beta, \sigma_\alpha^2$, and $\sigma_e^2$.  

While \cite{e2008} were able to derive well-performing generalized fiducial intervals for $\sigma_\alpha^2$ and $\sigma_e^2$, they did so by way of a long calculation and they did not allow for simultaneous inference on the fixed effects $\beta$.  Using \eqref{eq:Jacobian} the generalized fiducial intervals can be implemented without major computational hassle using the STAN software, while simultaneously obtaining GFD for the fixed effects $\beta$ that is often of interest \citep{NeupertEtAl2020}.

Due to the structure imposed by  \eqref{eq:MMDGA}, the Jacobian matrix simplifies greatly and does not involve the square root of a complicated matrix.  Indeed, the function $J(\mathbf{Y}, S_\alpha, \sigma_\alpha^2,\sigma_e^2, \xx \beta)$ is computed from the $n\times (d+2)$ matrix obtained by column concatenation of
\[
  \nabla_\beta \textbf{Y} = \xx,\quad
  \frac{\partial \textbf{Y}}{\partial \sigma_e^2} = (\sigma_\alpha^2 S_\alpha+ \sigma_e^2 I)(\textbf{Y} - \xx \beta),\quad
  \frac{\partial \textbf{Y}}{\partial \sigma_\alpha^2} = S_\alpha (\sigma_\alpha^2 S_\alpha+ \sigma_e^2 I)(\textbf{Y} - \xx \beta).
\]
Derivation of the above quantities is outlined in Appendix \ref{a:MM}.  

We will draw samples from our GFD to create fiducial confidence intervals that we will then compare to the truth to assess performance, parallel the simulation study in \cite{e2008}. In particular, at each instance of our simulation, we generated data using seven different designs of group sizes outlined in Table \ref{tab:mmgroups} and one of the true parameter values $(\sigma_\alpha^2, \sigma_e^2)$ from the set $\{(.1,10),(.5,10),(1,10),(.5,2),(1,1),(2,.5),(5,.2),(10,.1)\}.$

\begin{table}
\centering
\begin{minipage}{.9\textwidth}
  \centering
  \begin{center}
 \begin{tabular}{||c c c c c c c c c||} 
 \hline
 Pattern & $\Phi$ &~& & & $n_i$ & & &  \\
 \hline
 \hline
 1 & .068 &~& 1 & 1 & 1 & 1 & 1 & 100 \\
 \hline
 2 & .130 &~& 2 & 2 & 2 & 2 & 2 & 100 \\
 \hline
 3 & .187 &~& 2 & 5 & 60 & & &  \\
 \hline
 4 & .410 &~& 4 & 4 & 4 & 8 & 48 &  \\
 \hline
 5 & .700 &~& 5 & 10 & 15 & 20 & 25 & 30  \\
 \hline
 6 & .807 &~& 2 & 2 & 4 & 6 & &  \\
 \hline
 7 & .957 &~& 6 & 6 & 8 & 8 & 10 & 10 \\
 \hline
\end{tabular}
\end{center}
  \captionof{table}{The different group assignments tested for simulation of mixed models, as suggested by \cite{e2008}.  Here, $\Phi$ is an expression of imbalance as defined by \cite{ahrens1981}.  The smaller the value of $\Phi$, the greater the degree of imbalance.}
  \label{tab:mmgroups} 
\end{minipage}
\centering
{\small \\} 
\end{table}

\begin{table}[ht]
  \checkoddpage
  \edef\side{\ifoddpage l\else r\fi}%
  \centering
  \makebox[\textwidth][c]{%
    \begin{minipage}[t]{1.3\textwidth}
    \centering
      \begin{tabular}{||c | c c c c c c c ||} 
 \hline
($\sigma^2_\alpha, \sigma^2_e) \backslash$ Pattern  & 1 & 2 & 3 & 4 & 5 & 6 & 7 \\
 \hline
 \hline
 (.1,10) & 34.06 (22.54) & 19.38 (13.44) & 88.16 (93.09) & 12.55 (10.49) & 4.00 (3.50) & 38.63 (39.64) & 7.06 (6.04) \\
 \hline
 (.5,10) & 35.84 (27.90) & 22.21 (16.12) & 102.62 (111.85) & 14.99 (12.67) & 5.92 (5.56) & 45.66 (44.97) & 9.01 (8.18) \\
 \hline
 (1,10) & 37.70 (30.91) & 23.89 (18.58) & 114.65 (123.98) & 18.77 (17.68) & 8.63 (8.59) & 49.28 (54.27) & 10.98 (9.79) \\
 \hline
 (.5,2) & 9.26 (7.26) & 6.18 (5.28) & 28.42 (34.99) & 5.38 (5.82) & 3.29 (2.85) & 12.88 (15.37) & 3.75 (3.23) \\
 \hline
 (1,1) & 7.62 (8.05) & 6.85 (6.34) & 33.84 (49.73) & 7.45 (8.04) & 5.58 (4.71) & 14.29 (18.18) & 5.61 (4.95) \\
 \hline
 (2,.5) & 11.33 (11.00) & 10.74 (10.05) & 52.34 (87.86) & 14.83 (15.00) & 10.18 (9.08) & 24.27 (28.87) & 10.02 (9.80) \\
 \hline
 (5,.2) & 24.47 (24.48) & 26.08 (23.21) & 140.51 (213.16) & 34.27 (34.80) & 24.96 (22.16) & 54.17 (64.35) & 24.79 (22.60) \\
 \hline
 (10,.1) & 49.16 (43.76) & 47.13 (44.41) & 279.18 (419) & 69.09 (67.37) & 52.37 (42.99) &  102.33 (135.72) & 47.55 (40.69) \\
 \hline
 
\end{tabular} \\
\caption{{\myfont Median (IQR) of the CI lengths for all our parameters.  Instances of outstanding lengths likely contribute to the conservative empirical coverage we see in Figure \ref{fig:RE_empCov}.  These results are comparable to the conclusions in \cite{e2008}.} }\label{tab:RE_CILen}
    \end{minipage}%
  }%
\end{table}

\begin{figure}[ht] 
  \begin{minipage}[b]{0.44\linewidth}
    \centering
    \includegraphics[width=\linewidth]{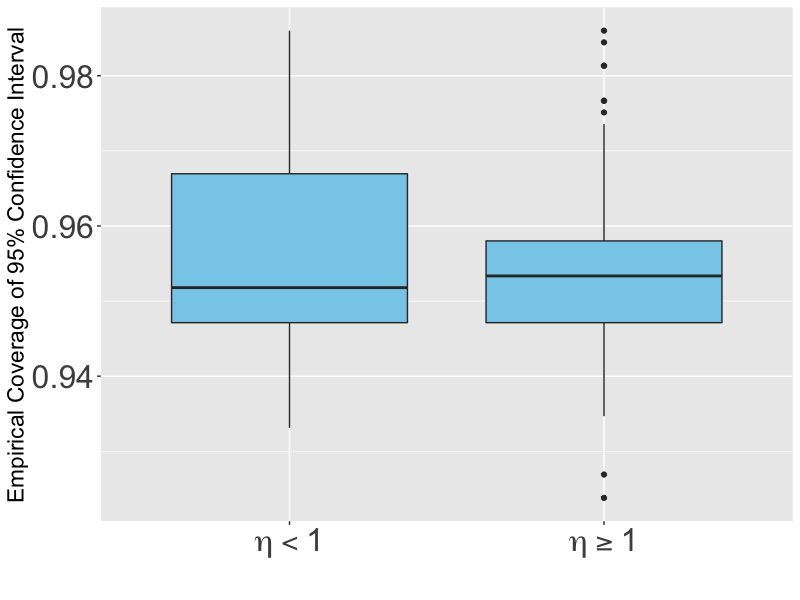} 
    \caption{Empirical coverages for $\eta = \sigma^2_\alpha/\sigma_e <1$ and $\eta \geq 1$ for both $\sigma_\alpha^2$ and $\sigma_e^2$ across all parameter choices {\myfont and all seven group choices.}  We can see that the generalized fiducial method tends to be more conservative for larger values of $\eta$.} \label{fig:RE_empCov}
  \end{minipage}\hspace{1mm}
\end{figure}

The results of our simulation can be seen in Figures~\ref{fig:RE_empCov} and {\myfont Table \ref{tab:RE_CILen}}.  Across many different parameter choices, it would appear that the generalized fiducial method performs well, with a slight conservative tendency for certain extreme parameter combinations.  Our findings match \cite{e2008}, who found this same conservative tendency for larger values of $\eta = \sigma_\alpha^2/ \sigma_e^2$ and greater degrees of imbalance, yet promising performance in all other instances.
We conclude by pointing out that we were able to efficiently implement our calculations using RStan \citep{STAN} which yields accurate results with a frugal use of iterations.


\section{Binomial Distribution}
\subsection{Binomial Distribution; \texorpdfstring{$n$}{} known, \texorpdfstring{$p$}{} unknown}\label{s:BinomialP} When working with discrete distributions, Equation \eqref{eq:Jacobian} can no longer be used, and one must instead invert the DGA directly.  To demonstrate how one applies the generalized fiducial method to discrete data, we begin with a well studied application of the GFI ideas: inference on $p$ from the binomial distribution with $n$ known. 
The following is a possible data generating algorithm
\begin{equation}\label{eq:BinDGA}
 Y=\sum_{i=1}^n I_{(0,p)}(U_i),
\end{equation}
where $p\in [0,1]$ is the unknown parameter and $U_i$ are i.i.d.~Uniform$(0,1)$. 


Assume we have an observed realization of our random variable $Y = y$. The uniform random variables $U_i$ that generated the observation $y$ are unknown and therefore we replace them with newly generated $U_i^*$. When we solve \eqref{eq:BinDGA} for $p$, we obtain the interval $(U_{(y)}^*,U_{(y+1)}^*]$, where $U_{(y)}^*$ is the $y$th order statistic of $U^*_i$. Recall that $U_{(y)}^*$ follows the $Beta(y,n-y+1)$ distribution.

This illustrates an important feature of GFI for discrete data: it is based on a random set rather than a random point.  {\myfont This is in contrast to a GFD derived for continuous data using Theorem \ref{Jacobian}.}  Notions from Dempster-Shafer theory \citep{dempster, edlefsen2009}, such as \textit{belief} and \textit{plausibility}, can be used to interpret these sets. Another approach is to reduce the interval to a single  distribution. Typically this is defined as function of the densities of the upper and lower bound of the interval, the density of $Beta(y,n-y+1)$ and $Beta(y+1,n-y)$. \cite{hannig2009AP} suggest using the arithmetic mean, while \cite{SchwederHjortBook} advocate for the geometric mean.  {\myfont Let $f_{\beta(a,b)}$ be the PDF of a $Beta(a,b)$ distribution.  The arithmetic mean leads to $r_{y,n}(p) \propto \left( p \left(n -2y \right) + y \right) f_{\beta(y, n-y)}(p)$ for $y\neq 0,n$ while the geometric mean leads to $r_{y,n}(p) \propto f_{\beta\left(y + \frac{1}{2}, n - y + \frac{1}{2} \right)}(p)$.  
Thus, simulating values from these GFDs is computationally easy and effective, as they are each derived from Beta density functions.  }

This result immediately generalizes to multiple independent observations from a binomial distribution.  Let ${Y}_1, \dots, {Y}_m$ be i.i.d. $Bin(n,p)$ where $n$ is known but $p$ is unknown.  The sufficient statistic $\sum_{i=1}^m Y_i$ follows  $Bin(n m, p)$ and we can use \eqref{eq:BinDGA} to obtain a GFD.

\subsection{Binomial Distribution; \texorpdfstring{$n$}{} unknown, \texorpdfstring{$p$}{} known}
\label{s:BinomialN}
With the aim of providing a more nuanced perspective on the generalized fiducial approach for discrete data, we will now consider a more complicated problem: binomial data for $n$ unknown. To the best of our knowledge, the binomial distribution with unknown $n$ has yet to receive a fiducial solution in the literature.

Consider $m$ i.i.d. observations from a $Bin(n,p)$ distribution, $\textbf{Y} = ({Y}_1, \ldots, {Y}_m)$  where $p$ is {\em known} but $n$ is {unknown}, and our objective is inference on $n$.  We begin our solution by choosing a DGA based on the inverse of the distribution function.  In particular,
\begin{equation} \label{eq:BinNDGA}
{Y}_i = F^{-1}_{n,p} ({U}_i),\quad i=1,\ldots,m,
\end{equation}
where ${U}_1,\ldots, U_m$ are i.i.d. $ Unif(0,1)$, and $F^{-1}_{n,p}(u)=\inf\{y: F_{n,p}(y)\geq u\}$ with $F_{n,p}$ being the distribution function of $Bin(n,p)$.

There are two major observations to make about inverting \eqref{eq:BinNDGA}: first, that $n \geq \max\{\mathbf{Y}\}$; and second, that $n$ does not have an upper bound.  Because of this latter observation, we must begin by approximating the GFD by restricting our calculations to a range of reasonable values {\myfont $\{\max\{\mathbf Y\}, \dots, N-1, N\}$ for some upper bound $N \in \nn$.} Then, the generalized fiducial probability mass of all relevant subsets in this range is determined by a deterministic algorithm.  That is, unlike the usual simulation based algorithms this algorithm directly assigns meaningful probabilities to sets of possible $n$ values. The precise details of this algorithm, {\myfont including our proposal for choosing the upper bound $N$,} can be found in Appendix \ref{a:BinN}.

We remark that the difference between GFD and Bayesian posterior using the flat prior is a result of the fiducial probability assigned to non-singleton sets, and if that probability is negligible the GFD would coincide with the Bayesian posterior.  In the simulations below, the final GFD on the values of $n$ is obtained by sampling from end-points of all of these sets that have positive fiducial probability. 


We illustrate the generalized fiducial solution to this inference problem via a simulation study, comparing the generalized fiducial solution against the Bayesian.  For this simulation, we generate data with 10, 50 and 100 observations from the distribution $Bin(n_0,p)$ where $n_0 = 10$ always and $p$ varies over the set $\{ 0.01, 0.05, 0.1, 0.2, 0.4, 0.5, 0.6, 0.7, 0.8, 0.9, 0.95, 0.99 \}$.  For each instance of our binomial data we draw 1000 values $\hat{n}_{B}$ from the Bayesian posterior distribution on $n$ as well as 1000 values $\hat{n}_{F}$ from the GFD on $n$.  Figure \ref{fig:errorBars} shows the $95\%$ GFD fiducial confidence intervals and Bayesian credible intervals based off of these 1000 draws.  This simulation reveals that while the generalized fiducial solution does capture the true value $n_0$ with an expected amount of regularity, there is not much distinguishing it from the Bayesian solution to the same problem.

\begin{figure}[ht]
  \checkoddpage
  \edef\side{\ifoddpage l\else r\fi}%
  \centering
  \makebox[\textwidth][c]{%
  \begin{minipage}{.59\textwidth}
  \centering
  \includegraphics[width=\linewidth]{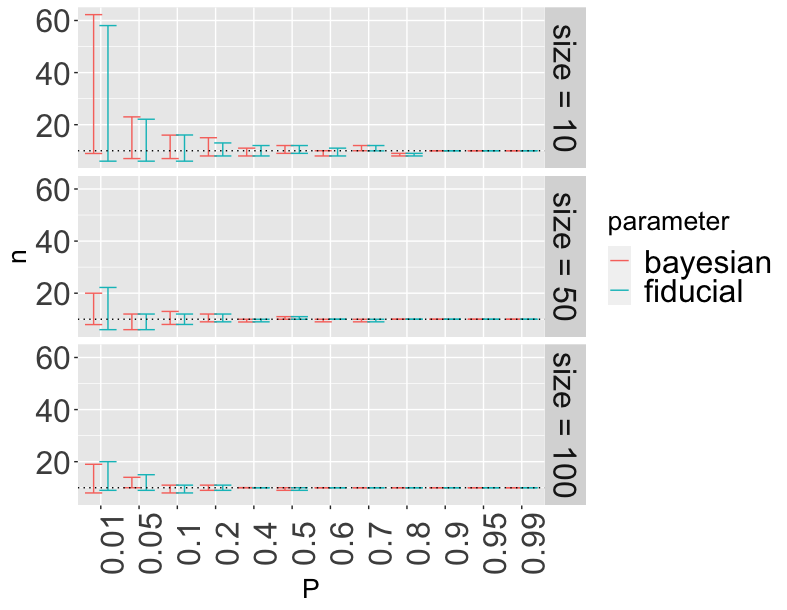}
  \vspace*{-0.5cm}
  \captionof{figure}{95\% fiducial confidence and credible intervals for each candidate value of $p$, over different data sizes.  {\myfont This simulation reveals that while the generalized fiducial solution does capture the true value $n_0$ with an expected amount of regularity, there is not much distinguishing it from the Bayesian solution to the same problem.}}
  \label{fig:errorBars} 
\end{minipage}%
    \hspace*{1cm}
  \begin{minipage}{.59\textwidth}
  \centering
  \includegraphics[width=\linewidth]{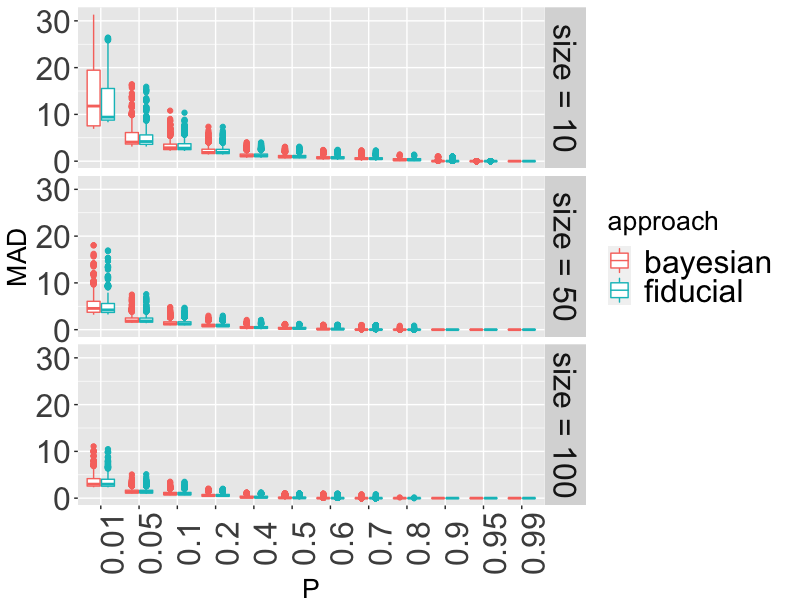}
  \vspace*{-0.5cm}
  \captionof{figure}{Boxplots of 1000 MAD calculations over varying values of $p$ and data sizes.  {\myfont There is some slight variation between the fiducial and Bayesian solutions for the smaller data sizes with smaller values of $p$.}}
  \label{fig:meanDiff} 
\end{minipage} %
  }%
\end{figure}

To better examine the difference between the generalized fiducial and the Bayesian solutions to this problem, we perform the above simulation 1000 times: each time drawing 1000 values from the GFD and 1000 values from the Bayesian posterior distribution.  At each repetition, we record the Mean Absolute Difference (MAD) $\frac{1}{1000} \sum_{i = 1}^{1000} | n_0 - \hat{n}_{\eta, i} |$ for both the generalized fiducial and the Bayesian draws, where $\eta \in \{\text{Fiducial}, \text{ Bayesian}\}$. 

Figure \ref{fig:meanDiff} compares boxplots of each method's 1000 MAD values collected via the described simulation.  {\myfont We can see from this simulation that the generalized fiducial and Bayesian solutions are mostly identical, with some slight variation when the data sizes and values of $p$ are small.  We suspect that this slight variation is due to the fact that, in our calculations for these cases, the fiducial method assigned more mass to \textit{sets} of values.  This aligns well with the fiducial method's major contrast to the Bayesian method: the former can assign mass to sets of $n$ values while the later only assigns mass to singleton values.  }

\begin{figure}
\centering
\begin{minipage}{.6\textwidth}
  \centering
  \includegraphics[width=\linewidth]{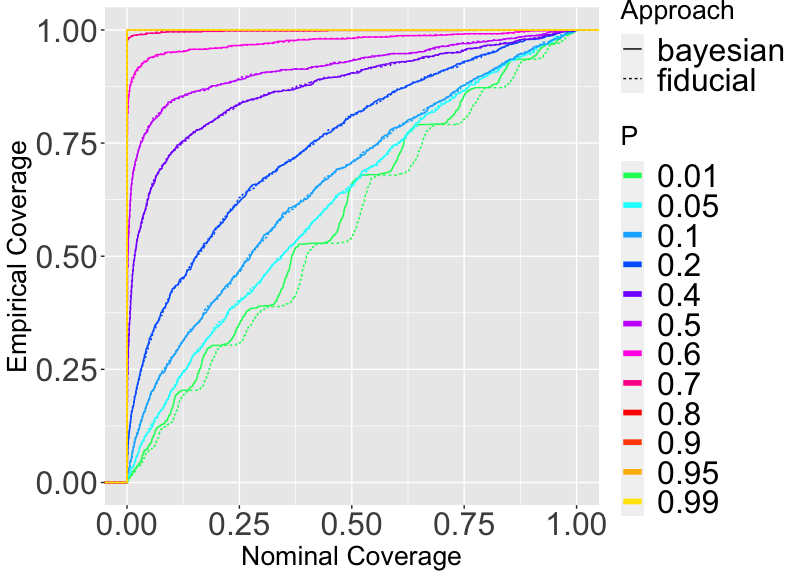}
  \captionof{figure}{The coverage of fiducial confidence and credible intervals at varying levels of confidence for the generalized fiducial and Bayesian approaches. This figure was calculated in the same way as \textbf{Fig. 1}.  {\myfont The Bayesian solution matches the generalized fiducial, except for some slight variation when $p$ is $0.01$ and the data size is small.}}
  \label{fig:UpperP} 
\end{minipage}
\centering
{\small \\} 
\end{figure}

Figure \ref{fig:UpperP} examines the approximate fiducial confidence intervals, built with the upper 95\% of the GFD mass, and the Bayesian upper 95\% credible intervals, of 1000 simulations of a size 100 draw from a $Bin(10,\cdot)$ distribution.  The horizontal axis of Figure \ref{fig:UpperP} is the coverage probability we set for a given interval based off of the posterior distribution we derive, while the vertical axis is the observed proportion of time that a given interval captured the true parameter.

Figure \ref{fig:UpperP} reaffirms what we saw in Figures \ref{fig:errorBars} \& \ref{fig:meanDiff}.  {\myfont Namely, that the Bayesian solution matches the generalized fiducial, except for some slight variation when $p$ is $0.01$ and the data size is small}.  We also observe that for larger values of $p$ an interval at any coverage level will contain the true parameter value.  In the case of $p = 0.99$, an interval at any non-zero confidence level captures the true parameter with complete certainty.  
This relationship between nominal and empirical coverage for larger values of $p$ follows from what we observe in Figure \ref{fig:errorBars}, which suggests that our posterior distribution is just a single point (the true parameter value $n_0$), and thus intervals of at any confidence level are simply the set $\{ n_0 \}$.  
For smaller values of $p$ the coverage plots are closer to the 45 degree line indicating a close to correct coverage. The waves in the plot are caused by discretization effects.

Our results show that inference on $n$ for known $p$ is fairly straightforward when the probability of success is high.  For smaller values of $p$, while both the generalized fiducial and Bayesian approaches have less accuracy, they remain comparable.  

\subsection{Binomial Distribution; \texorpdfstring{$n$}{} unknown, \texorpdfstring{$p$}{} unknown}
\label{s:Binom ialNP}
We will now consider our final non-trivial example through the fiducial lens: simultaneous inference on $(n,p)$ for the binomial distribution.  As far as we are aware, this problem has yet to receive a satisfactory solution using any philosophy whatsoever.  

We consider a set of $m$ independent binomial data: $\textbf{Y} = (Y_1, \dots, Y_m)$ i.i.d.~$Bin(n,p)$, where both $p$ and $n$ are unknown and our objective is inference on the \textit{pair} $(n,p)$.  We will use the same DGA \eqref{eq:BinNDGA} as in the previous example. Recall the following identity relating the binomial CDF $F$ and the beta CDF $G$ such that $F_{n, p}(y) = G_{n - y + 1, y}(1-p)$ \citep{wadsworth}.  This allows us to write an equivalent form of \eqref{eq:BinNDGA}:
\begin{equation}\label{eq:DGEbeta}
G_{Y_{i} + 1, n - Y_{i}}^{-1}(1 - {U}_{i}) \geq p > G_{Y_{i} , n - Y_{i} + 1}^{-1}(1 - {U}_{i}),\quad i=1,\ldots,m,
\end{equation}
which we can invert to get
\[ \hat{p}_{i,n}^{upper} = G_{Y_{i} + 1, n - Y_{i}}^{-1}(1 - U_{i}),\quad \hat{p}_{i,n}^{lower} = G_{Y_{i} , n - Y_{i} + 1}^{-1}(1 - U_{i}), \quad i = 1, \ldots, m.  \]
For any value $\hat{n}$, the set of associated values $p$ that satisfy \eqref{eq:DGEbeta}
is the interval 
\[(\max\{\hat{p}_{1,\hat{n}}^{lower},\ldots,\hat{p}_{m,\hat{n}}^{lower}\},\min\{\hat{p}_{1,\hat{n}}^{upper},\ldots,\hat{p}_{m,\hat{n}}^{upper}\}].
\]
The interval is taken as an empty set if the lower limit is larger than the upper limit.  
\begin{figure}
  \begin{minipage}{\textwidth}
  \centering
  \includegraphics[width=0.45\linewidth]{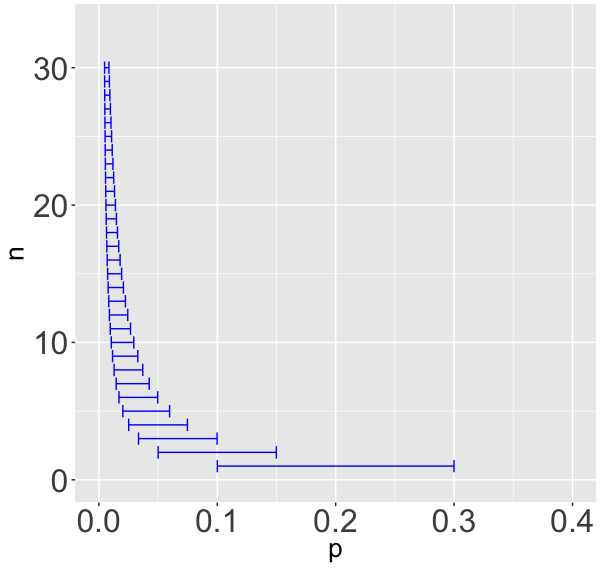}
\hspace{1ex}
  \includegraphics[width=0.45\linewidth]{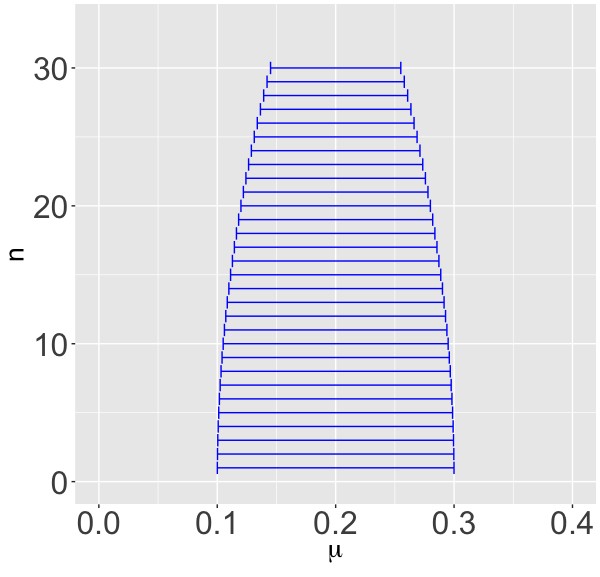}
  \captionof{figure}{Two parameterizations for the same single draw from the GFD: using $n,p$ parameterization (on the left) versus using $n,\mu$ parameterization (on the right). Note that a draw from this GFD defines a region of values in a joint parameter space that is discrete in $n$ and continuous in $p$ and $\mu$.}
  \label{fig:GFDNPparameterization} 
\end{minipage} 
\end{figure}

A generalized fiducial solution to this problem would define a GFD that assigns probability mass to \textit{sets} 
\[\bigcup_{\hat{n}}\{\hat{n}\} \times (\max\{\hat{p}_{1,\hat{n}}^{lower},\ldots,\hat{p}_{m,\hat{n}}^{lower}\},\min\{\hat{p}_{1,\hat{n}}^{upper},\ldots,\hat{p}_{m,\hat{n}}^{upper}\}]. \] These sets contain multiple predicted values $\hat{n}$ that are each associated with a set of values of $p$.  The left panel of Figure \ref{fig:GFDNPparameterization} shows a single random draw from the GFD on the parameter space $(n,p)$.  Note that a single draw in the two-dimensional parameter space is a collection of horizontal segments at each $\hat{n}$.  Naturally, given a fixed data set, $Y$, larger and larger values of $\hat{n}$ correspond with smaller values for $\hat{p}$, making them difficult to visualize.  For this reason, and because of some desirable properties of the limiting distribution of $\hat{n}\hat{p}$ discussed in Appendix \ref{a:BinNP}, we propose to work with a reparameterization $(n,\mu)$, where $\mu := np$.  This results in a GFD that assigns mass to sets of the form \[\bigcup_{\hat{n}} \{\hat{n}\} \times (\max\{\hat{\mu}_{1,\hat{n}}^{lower},\ldots,\hat{\mu}_{m,\hat{n}}^{lower}\},\min\{\hat{\mu}_{1,\hat{n}}^{upper},\ldots,\hat{\mu}_{m,\hat{n}}^{upper}\}],\] where
\[  \hat{\mu}_{i,\hat{n}}^{upper} = \hat{p}_{i,\hat{n}}^{upper} \hat{n},\quad \hat{\mu}_{i,\hat{n}}^{lower} = \hat{p}_{i,\hat{n}}^{lower} \hat{n}, \quad i = 1, \ldots, m. \]
A draw from this reparameterized GFD can be seen in the right panel of Figure \ref{fig:GFDNPparameterization}.


Our algorithm for simulating values from the GFD of $(n,\mu)$ uses a Metropolis within Gibbs sampler.  We start with a set of values $\mathbf{U}$ that give a solution that is not the empty set, we record this solution set, then we resample the $\mathbf{U}$ values one by one such a way that the DGA \eqref{eq:BinNDGA} is still satisfied. These updates are available in closed form.  To improve mixing we also add to each Gibbs scan two additional random walk Metropolis Hastings steps, one in the $n$ direction and one in the $\mu$ direction, so that we can adequately explore our sample space.  Further details on the MCMC algorithm can be found in Appendix \ref{a:BinNP}.


\begin{figure}[ht]
  \checkoddpage
  \edef\side{\ifoddpage l\else r\fi}%
  \centering
  \makebox[\textwidth][c]{%
    \begin{minipage}[t]{1.35\textwidth}
      \centering
    \includegraphics[width=0.08\linewidth]{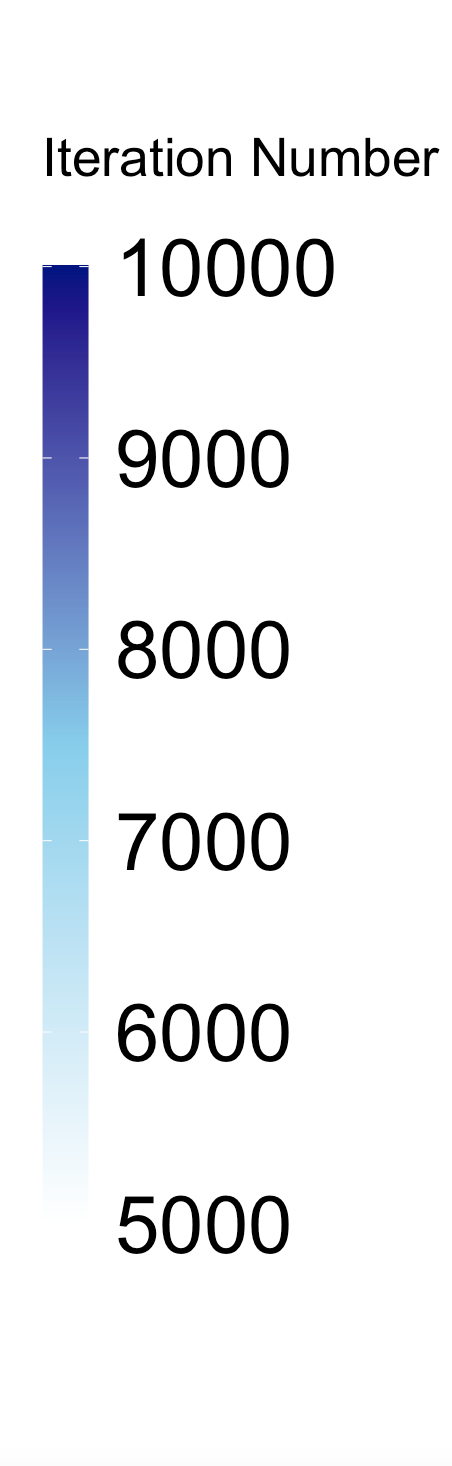} 
    \includegraphics[width=0.45\linewidth]{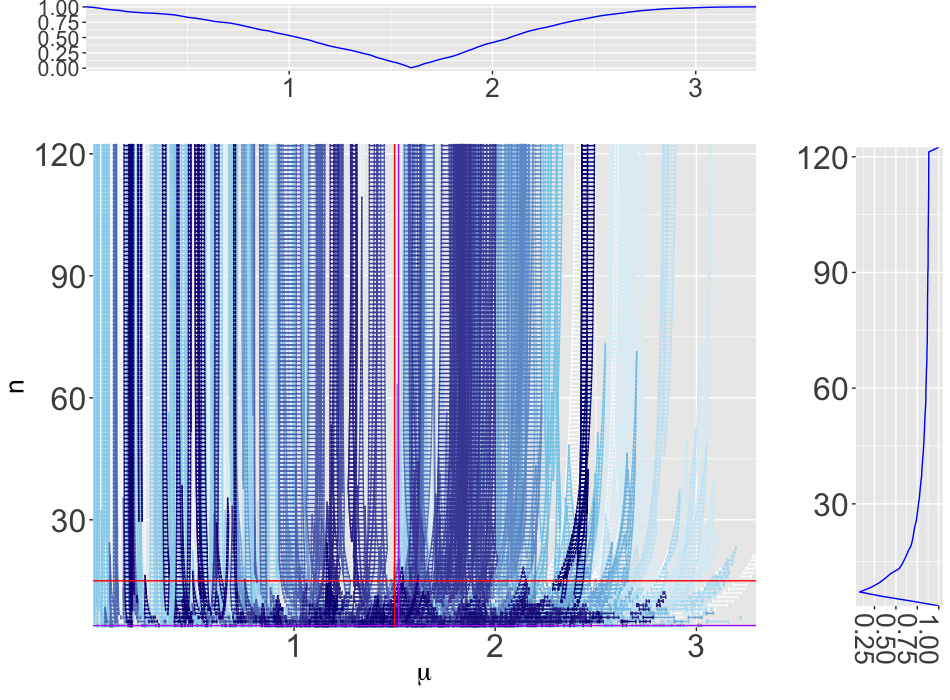} 
    \includegraphics[width=0.45\linewidth]{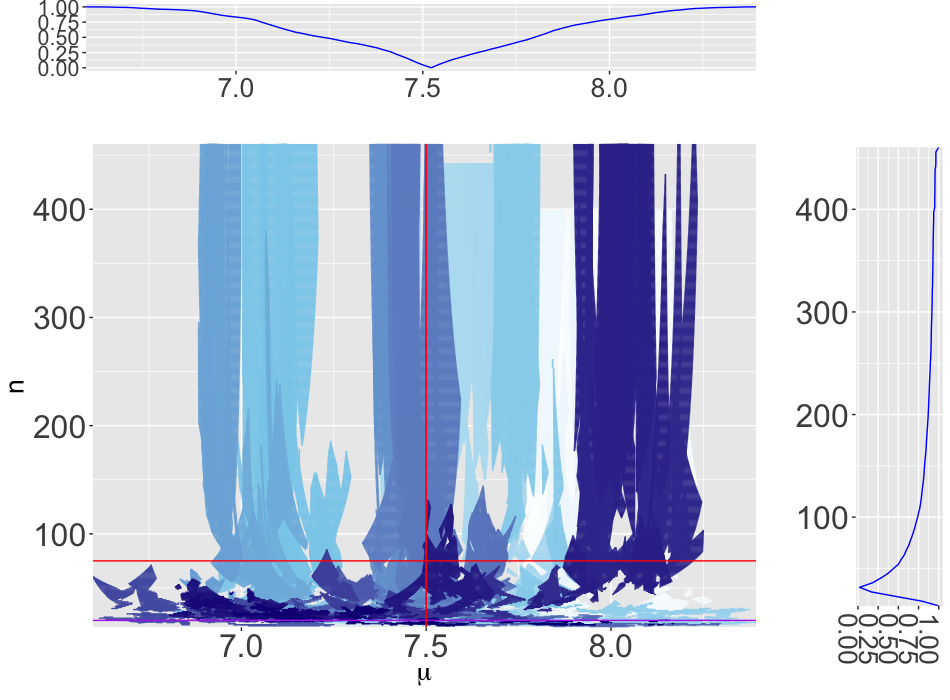}
    \caption{GFD for $Bin(15,0.1)$ data (left) and $Bin(75,0.1)$ data (right).  Different shades of color in the center plot correspond to separate draws using a Gibbs sampler on the GFD built on the data.  The top and right marginal plots show the confidence curves of $\mu$ and $n$ values, respectively, centered at the median.  The red lines cross at the true $(n,\mu)$ value used to generate the distribution and the purple lines cross at a value of \cite{DasGupta} estimator  $(\hat{n},\bar{\mathbf{Y}})$ that are used to initialize our algorithm.} 
    \label{fig:nprun1}
    \end{minipage}%
    }%
\end{figure}
\begin{figure}[ht]
  \checkoddpage
  \edef\side{\ifoddpage l\else r\fi}%
  \centering
  \makebox[\textwidth][c]{%
    \newline
    \begin{minipage}[t]{1.35\textwidth}
      \centering
    \includegraphics[width=0.08\linewidth]{color_scale.png} 
    \includegraphics[width=0.45\linewidth]{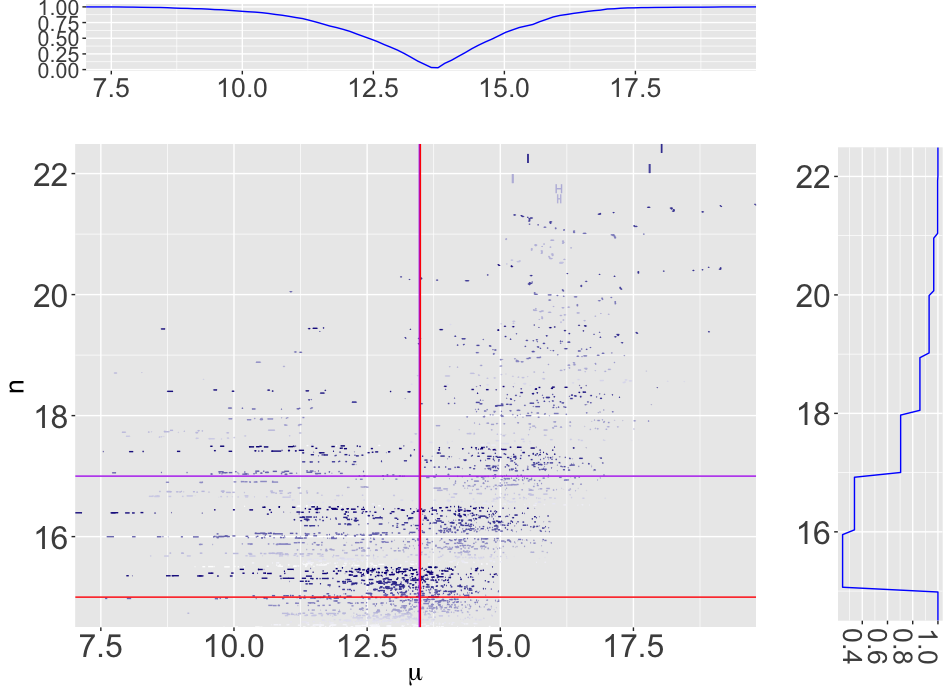} 
    \includegraphics[width=0.45\linewidth]{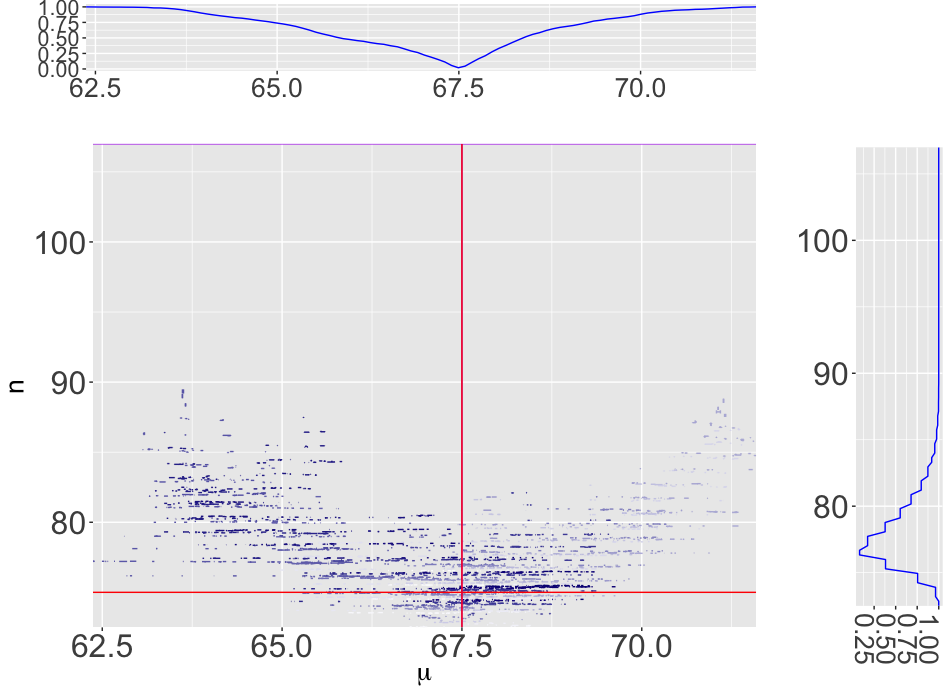}
    \caption{GFD for $Bin(15,0.9)$ data (left) and $Bin(75,0.9)$ (right) visualized in the same way as Figure \ref{fig:nprun1}.  {\myfont For the $Bin(15,0.9)$ data, the values are jittered in the $n$ direction for clarity.}} 
    \label{fig:nprun2}
    \end{minipage}%
  }%
\end{figure}

\begin{figure}[ht]
  \checkoddpage
  \edef\side{\ifoddpage l\else r\fi}%
  \centering
  \makebox[\textwidth][c]{%
    \begin{minipage}[t]{1.35\textwidth}
      \centering
    \includegraphics[width=0.45\linewidth]{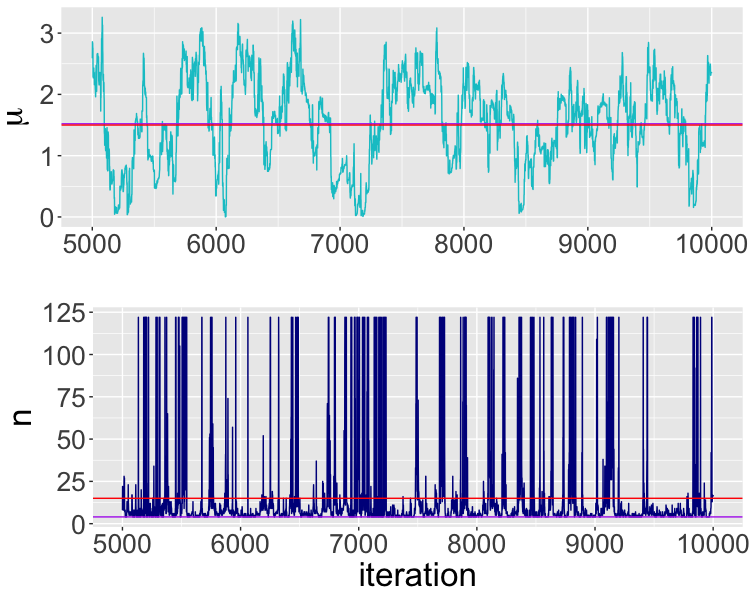} \includegraphics[width=0.45\linewidth]{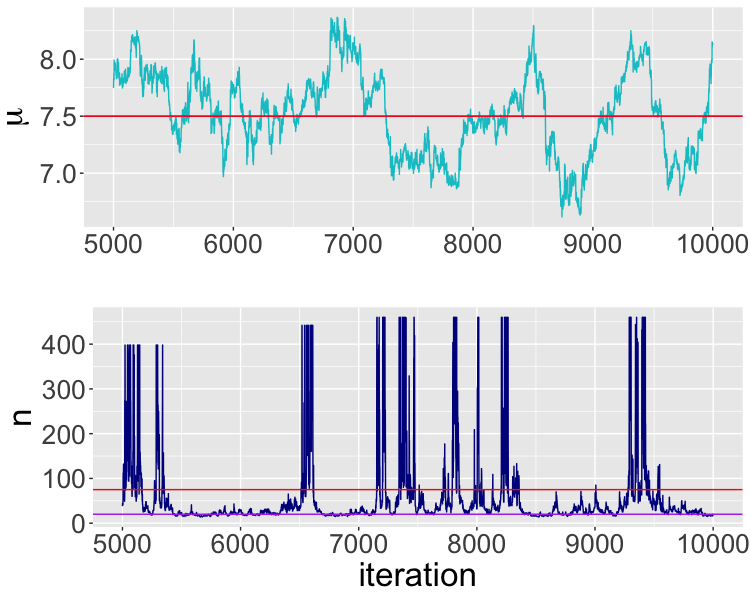}
    \caption{Marginal trace plots for $\mu$ and $n$ for $Bin(15,0.1)$ data (left) and $Bin(75,0.1)$ data (right).  These plots use the same GFD produced for figures \ref{fig:nprun1} and \ref{fig:nprun2}.  The red line crosses at the true parameter values used to generate the data and the purple line crosses at a value of \cite{DasGupta} estimator  $(\hat{n},\bar{\mathbf{Y}})$ that are used to initialize our algorithm.} 
    \label{fig:nprun1_trace}
    \end{minipage}%
    }%
\end{figure}

\begin{figure}[ht]
  \checkoddpage
  \edef\side{\ifoddpage l\else r\fi}%
  \centering
  \makebox[\textwidth][c]{%
    \newline
    \begin{minipage}[t]{1.35\textwidth}
      \centering
    \includegraphics[width=0.45\linewidth]{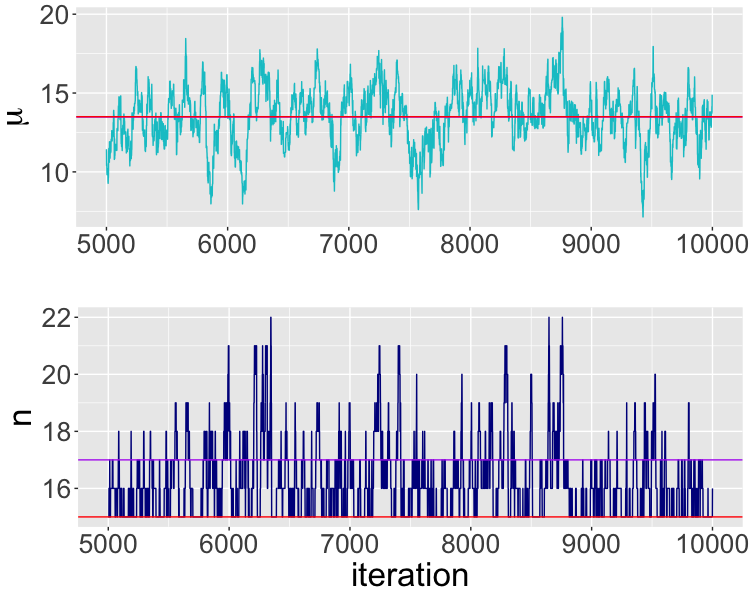} 
    \includegraphics[width=0.45\linewidth]{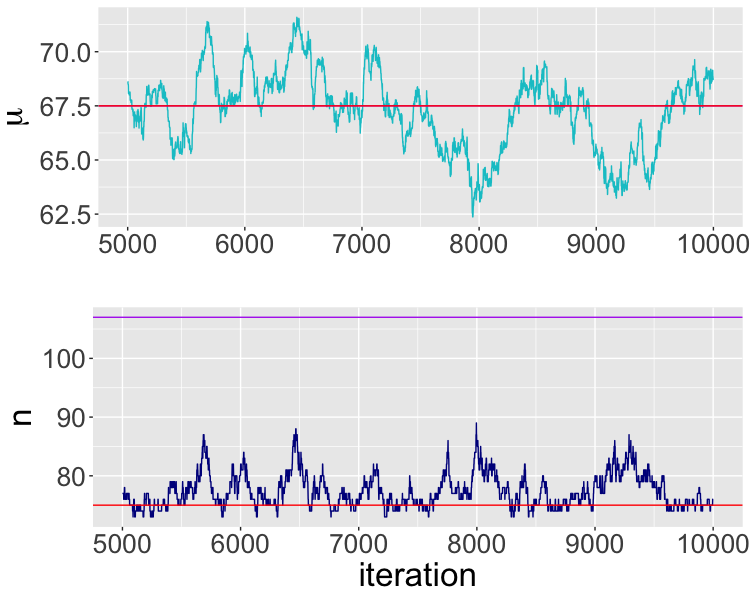}
    \caption{Marginal trace plots for $\mu$ and $n$ for $Bin(15,0.9)$ data (left) and $Bin(75,0.9)$ (right) visualized in the same way as Figure \ref{fig:nprun1_trace}.} 
    \label{fig:nprun2_trace}
    \end{minipage}%
  }%
\end{figure}

Figures \ref{fig:nprun1} \& \ref{fig:nprun2} show four implementations of our Gibbs sampler for different true values of $n$ and $p$.  For these figures we draw 50 observations from the distributions $Bin(0.1,15)$, $Bin(0.1,75)$, $Bin(0.9,15)$, and $Bin(0.9,75)$, respectively.   


For $p$ small, the binomial distribution approaches a Poisson distribution and results in large sets of $\hat{n}$ values for each draw from our GFD.  One can observe this in Figure \ref{fig:nprun1}, where certain draws from the GFD seem to go upward to infinity in the $n$ direction.  Our algorithm senses when a draw from the GFD is reaching its limiting Poisson distribution and stops checking new values for $n$.  The specific details of this stopping criterion are discussed in Appendix \ref{a:BinNP}. In contrast, for large $p$, the algorithm reports a narrower range in $n$ (Figure \ref{fig:nprun2}).  Note that for clarity we jittered the graph in the $n$ direction, especially in the case of Figure \ref{fig:nprun2}, where most of the sets have a $n$ value concentrated in one area. 

Our Gibbs sampler is able to explore both the $n$ space and the $\mu$ space simultaneously.  We examine the marginal movement in the $\mu$ and the $n$ directions in Figures \ref{fig:nprun1_trace} \& \ref{fig:nprun2_trace}.  To create these figures, we select a single representative $\mu$ and $n$ from each draw from the GFD.  To select a single $\mu$ from a set of values in a single draw from the GFD, we randomly choose either the minimum or maximum $\mu$ value.  We perform the same process in the $n$ direction; for every draw from the GFD, we randomly select either the largest or smallest $n$ value present to be that draw's representative in the trace plot.  For each of these selections, the minimum or maximum are chosen with equal probability.  The resulting figures show how our algorithm explores the $n$ and $\mu$ space around the true parameter values.  For smaller true values of $p$, these plots show how our algorithm visits very large values of $n$.  These draws correspond to the instances where the algorithm begins to capture the limiting Poisson distribution.



We evaluate this method via a simulation study.  For each combination $n \in \{15, 75\}$ and $p \in \{0.1, 0.5, 0.9\}$, we drew $100$ observations from a Bin$(n, p)$ distribution 300 times.  One important thing to note about this solution is that each element of the GFD is a \textit{set} of values, rather than a singular pair $(\hat{n}, \hat{\mu})$.  This complicates the notion of containment ratios slightly.  Ideally, for each iteration of our simulation study, we would like to be able to calculate a box of reasonable values that we would expect to contain the true parameter pair $(n,\mu)$ 95\% of the time.  {\myfont We followed the Dempster-Shafer theory \citep{dempster} to draw belief and plausibility boxes, which we then used to calculate coverage.  In this application, we call belief box a rectangle centered at the median of the each marginal fiducial distribution that fully contained 95\% of our fiducial draws, while the plausibility box is a rectangle intersecting 95\% of our fiducial draws.  One can see the resulting containment ratios using belief and plausibility in Table \ref{tab:beliefplaus}.  We point out that as expected belief box is larger than plausibility box.  This is a reflection of the fact that in Dempster-Shafer theory belief of a set is smaller than plausibility of the same set.} 

\begin{table}
\centering
\begin{minipage}{.9\textwidth}
  \centering
  \begin{center}
 \begin{tabular}{||c | c | c | c | c ||} 
 \hline
 Distribution & \makecell{95\% plausibility \\ Containment} & \makecell{95\% Belief \\Containment}  & \makecell{$\mu$ Marginal \\ Coverage}  & \makecell{$n$ Marginal \\ Coverage}\\
 \hline
 $Bin(0.1,15)$ & 1 & 1 & 0.992 & 1 \\
 $Bin(0.1,75)$ & 1 & 1 & 0.943 & 1 \\
 $Bin(0.5,15)$ & 0.994 & 0.994 & 0.961 & 0.968 \\
 $Bin(0.5,75)$ & 0.955 & 0.955 & 0.961 & 0.968 \\
 $Bin(0.9,15)$ & 1 & 1 & 1 & 1\\
 $Bin(0.9,75)$ & 0.978 & 0.986 & 1 & 1 \\
 \hline
\end{tabular}
\end{center}
  \captionof{table}{The coverage of 95\% belief and plausibility boxes on posterior fiducial sets of $(n,\mu)$.}
  \label{tab:beliefplaus} 
\end{minipage}
\centering
{\small \\} 
\end{table}

{\myfont Table \ref{tab:beliefplaus} shows that our method performs mostly conservatively for our choices of parameter values.  This is in line with our expectations based on what we observed in Figures \ref{fig:nprun1} \& \ref{fig:nprun2}. The belief boxes for these draws from the fiducial posterior would need to be very large, and potentially unbounded, since they need to fully contain 95\% of the fiducial samples.}

\section{Conclusion and Discussion}

GFI provides an alternative perspective on numerous classical inferential problems.  Our selection of examples show how the generalized fiducial framework can derive a meaningful, practically feasible distribution on a target parameter without a need for an arbitrarily defined prior.  While there is much to be learned from using this framework on classic theoretical problems, modern research has shown that the fiducial perspective is a strong tool for tackling today's statistical challenges.  We provide a short list of such references here for the interested reader: \citep{e2008, cisewski2012,wandler2011,wandler2012a,wandler2012b,glagovskiy2006,hannig2009,liu2016}.

{\myfont While the fiducial method has seen much recent success, it has the potential drawback that inverting the DGA for some data applications can be non-trivial.  The result by \cite{GenFid}, Theorem \ref{Jacobian}, mitigates this issue by developing a general formula the GFD for continuous data.  Of course, Theorem \ref{Jacobian} replaces the need for a potentially complicated inversion by a need for a potentially complicated differentiation.  As a future step towards making fiducial inference accessible to a wide audience, the authors would like to see an autodifferentiator developed to automate the process of inverting a DGA.  With such a tool, the fiducial process could be made significantly more accessible.}

\section{Acknowledgements}
Jan Hannig's research was supported in part by the National Science Foundation under Grant No. IIS-1633074 and DMS-1916115.

\begin{appendix}



\section{Derivation of the Multivariate Jacobian Quantity \texorpdfstring{$J(X ; \mu, U, \Lambda)$}{} and Marginal \texorpdfstring{$r(\text{veck}(A))$}{}}\label{a:MVN}
Let $J^{i,j}$ be a matrix of all zeros save for a value of 1 at the index $(i,j)$.  In general, the form of our Jacobian derivative with respect to any value of the skew-symmetric matrix $A$ is $J(\mathbf{y},\mu, A,\Lambda) = \left[ \nabla_\mu \mathbf Y_i ; \frac{\partial\mathbf Y_i}{\partial \lambda_j}; \frac{\partial \mathbf Y_i}{\partial a_{j,k}} \right]$, where ``;" denotes the row concatenation of the following three quantities:

\begin{align*}
    \nabla_\mu \mathbf Y_i  &= I_d \\
    \frac{\partial\mathbf Y_i}{\partial \lambda_j} &= {\lambda_j}^{-1} (I_d -A)(I_d + A)^{-1} J^{j,j}(I_d +A)(I_d - A)^{-1}(\mathbf Y_i-\mu),\quad j=1,\ldots,d \\
    \frac{\partial \mathbf Y_i}{\partial a_{j,k}} &=  2(I_d+A)^{-1}\Big( -J^{i,j} + J^{j,i} \Big)(I_d-A)^{-1}(\mathbf Y_i-\mu),\quad 1\leq j<k\leq d.
\end{align*}

Using the $l_2$ norm, the final form of our Jacobian is $D(J(\mathbf{y},\mu, A,\Lambda))$ where $D(X) = \left| \det \left( \sqrt{X^TX}\right)\right|$.  One can show, using the Cauchy-Binnet formula, that the $\mu$ vector drops out and that we can factor out the $\Lambda$ matrix such that $D(J(\mathbf{y},\mu, A,\Lambda))  = \det(\Lambda)^{-1} D^*(J(\mathbf{y}, A))$.  {\myfont The notation $D^*(J(\mathbf{y}, A))$ is chosen to reiterate that, with the $\Lambda$ matrix factored out, the remaining Jacobian term relies on neither $\Lambda$ nor $\mu$.}  We can thus write the GFD likelihood using \eqref{eq:Jacobian} like so:
\begin{align*}
r_{\mathbf y}(\operatorname{veck}(A),\Lambda,\mu) \propto&  (2 \pi)^{\frac{-nd}{2}} |\Sigma|^{-n/2} \exp \left\{ \frac{-1}{2}  \operatorname{tr} \left( \Sigma^{-1} \left( \sum_{i=1}^n (\mathbf{y}_i - \mu)(\mathbf{y}_i - \mu)^T \right) \right) \right\} D(J).
\end{align*}
By integrating out $\mu$ and rearranging the terms that depend on $\lambda_i$ we get 
\[
r_{\mathbf y}(\operatorname{veck}(A),\Lambda) \propto D^*(J(\mathbf{y}, A)) (2 \pi )^{\frac{-d(n-1)}{2}} n^{-d/2} \left[ \prod_{i=1}^d \left(  \frac{1}{\lambda_i^2} \right)^{\frac{n}{2}} \exp \left\{  \frac{-1}{2\lambda_i^2} ( Z^T  nS^2 ~Z)_{ii} \right\}\right],
\]
Next, we integrate out $\Lambda$ to get the marginal GFD of $A$  
\[
r_{\mathbf y}(\operatorname{veck}(A)) 
\propto D^*(J(\mathbf{y}, A)) ( \pi )^{\frac{-d(n-1)}{2}} 2^{-d } n^{\frac{-d}{2}}  \Gamma\left(\frac{n-1}{2} \right)^d\prod_{i=1}^d( Z^T   n S^2 ~ Z)_{ii}^{\frac{-(n-1)}{2}}.
\]

A similar calculation can be done for an alternative data generating algorithm using minimal sufficient statistic $(\mathbf{\bar Y},S^2)$. The only difference will be the form of $D^*(J(\mathbf{y}, A))$. {\myfont As a reminder, the $\operatorname{veck}(A)$ operation vectorizes the strictly lower triangular elements of the skew-symmetric matrix $A$, as discussed in \cite{henderson1979}.}  See the STAN implementation on this paper's Github page for details.

\section{Derivation of the Mixed Models Jacobian Quantity \texorpdfstring{$J(\mathbf{Y}, S_\alpha, \xx\beta, \sigma_\alpha, \sigma_e)$}{}}\label{a:MM}
The derivative $\frac{\delta \textbf{Y}}{\delta \beta} = \xx$ is trivial.  To solve for the other two terms, note that $\Sigma = AA$ and therefore
\[
\frac{\delta \Sigma}{\delta\sigma_\alpha^2} = A \frac{\delta A}{\delta \sigma_\alpha^2} + \frac{\delta A}{\delta \sigma_\alpha^2}A, \mbox{ where }
\frac{\delta A}{\delta \sigma_\alpha^2} = \frac{1}{2} S_\alpha A^{-1}, 
\]
The derivative $\frac{\delta A}{\delta \sigma_\alpha^2}$ commutes with $A$ since
\[
    A = U\left(\text{diag}\left(\sqrt{\sigma_\alpha^2\lambda_i + \sigma_e^2}\right)\right)U',\quad
    \frac{\delta A}{\delta \sigma_\alpha^2} =\frac{1}{2} U\left(\text{diag}\left(\lambda_i \left( \sigma_\alpha^2\lambda_i + \sigma_e^2\right)^{-1/2}\right)\right)U',
\]
and clearly $\frac{\delta A}{\delta_\alpha^2} A = A\frac{\delta A}{\delta_\alpha^2}$.  Therefore,
\[
    \frac{\delta \textbf{Y}}{\delta \sigma_\alpha^2} = \frac{\delta }{\delta \sigma_\alpha} \left( \xx \beta + A \zz \right)\Big|_{\zz = A^{-1}*(Y - \xx\beta)} 
    =  \frac{1}{2} S_\alpha \Sigma^{-1} (\textbf{Y} - \xx \beta).
\]
{\myfont (Note that $A^{-1}(Y - \xx\beta)$ is a matrix product, not a realization of a DGA.)} Following this same logic gives \[\frac{\delta \textbf{Y}}{\delta \sigma_e} = \Sigma^{-1} (\textbf{Y} - \xx \beta) = (\sigma_\alpha^2 S_\alpha+ \sigma_e^2 I)^{-1} (\textbf{Y} - \xx \beta).\]

\section{Binomial algorithm for unknown \texorpdfstring{$n$}{}, known \texorpdfstring{$p$}{}}\label{a:BinN}
As in Section~\ref{s:BinomialP}, the generalized fiducial solution to the $n$ unknown problem assigns mass to \textit{sets} of integers rather than individual values.  Our aim is to assign fiducial probabilities to all reasonable sets of $n$ values in the sample space. 
Considering a \textit{set} $\mathbf{s}$ of candidate $n$ values and fixing the observed data $\mathbf{Y}$, we start by calculating the probability that we observe a set of uniforms values such that any $n^* \in \mathbf{s}$ could be the true $n$ value.  {\myfont Notice that the set of viable uniform values for any superset of $\mathbf{s}$ must necessarily be a subset of the set of viable uniform values for $\mathbf{s}$.}  The astute reader will notice that this is the same as the notion of {\it commonality}  found in Dempster-Schaffer calculus \citep{Shafer1976}.  We start by calculating the following inversion of \eqref{eq:BinNDGA}, 
\begin{equation}\label{eq:BinNDGAInversion}
F_{n^*,p}(Y_i - 1) < U_i^* \leq F_{n^*, p}(Y_i),\quad \text{ for }i=1,\ldots, n.
\end{equation}
Since $F_{n^*,p}(Y_i)$ is non-decreasing in $n^*$, we see that for any fixed $(U_1,\ldots,U_n^*)$ the set of $n^*$ is an interval of consecutive integers. We will denote the set of all integer intervals $\mathbf S=\{\{i,i+1, \ldots,j\}, 1\leq i\leq j\}$.  The commonality of the set $\mathbf{s}\in\mathbf S$ will be \[p(\mathbf{s} | \textbf{Y}) =\prod_{i=1}^m \left[F_{\max\{\textbf{s}\}, p}(y_{i}) - F_{\min\{\textbf{s}\}, p}(y_{i}-1)\right]^+.\]  

Our algorithm begins by using this commonality to address the issue of no strict upper bound on the set of possible values for $n^*$.  First, select a precision cutoff $\epsilon_1$.  Then, we define our set of candidate $n^*$ values, which we will call $\mathbf{N}$, by way of the following algorithm:  Starting with the observed maximum of the data  consider new values sequentially, comparing the commonality of the candidate $n^*$ value: $p(\{n^*\} | \textbf{Y})$ with commonalities observed so far. If the ratio of the commonality to the maximum commonality in $\mathbf{N}$ accepted so far
\[
\frac{p( \{n^*\} | \mathbf{Y})}{max_{k \in \mathbf{N}} \{ p(\{k\} | \textbf{Y})\}} > \epsilon_1,
\]
 add $n^*$ to our set $\mathbf{N}$ and consider $n^*+1$. Otherwise, halt the algorithm. 

Next we approximate the GFD by restricting ourselves only to subsets of $\mathbf{N}$. Define
$\mathbf{\hat S}= \{\mathbf s\in\mathbf S, \mathbf s\subset \mathbf N\}$.  
Using the analogue noted previously between Dempster-Shaffer calculus and the inversion of our DGA, we define the fiducial probability of a set $\mathbf{s}$ as the Dempster-Shaffer \textit{mass} of $\mathbf{s}$ \citep{Shafer1976, Yager2010} modified and renormalized so that $r_{\mathbf Y}(\emptyset)=0$.  
We calculate the the Dempster-Shaffer mass of $\mathbf{s}$ recursively by starting from $m(\mathbf{N})=p(\mathbf{N}|\mathbf Y)$ and then for 
$\mathbf s\in\mathbf{\hat S}$
\[ m( \textbf{s}) = \left[p(\mathbf{s} | \mathbf{Y}) - \left( \sum_{\textbf{r}\in\mathbf{\hat S} : \textbf{s} \subset \textbf{r}} m(r)   \right)\right]^{+}.\]
 Then, after the probability of all sets are defined in this way, we renormalize them so that they add to 1.  The fiducial probability for each element of $\mathbf s\in\mathbf{\hat S}$ is
\[ r_{\textbf{Y}}( \textbf{s}) = \frac{m(\mathbf s)}{1-m(\emptyset)}. \]


\section{Binomial algorithm for unknown \texorpdfstring{$n$}{}, unknown \texorpdfstring{$p$}{}}\label{a:BinNP}

We begin by making an important observation about the distribution of the predicted value $\hat{p}$.  For random $U_i^*$, fixed $\mathbf{Y}$ and $n$, we have that 
\[ \hat{p}_i^{upper} \sim Beta(Y_{i} + 1, n - Y_i),\quad
\hat{p}_i^{lower} \sim Beta(Y_{i} , n - Y_i + 1). \] 
Writing out the distribution on the bounds of our fiducial sets for $p$ in this way allows use of the following well-known result: \begin{lemma}\label{lem:stirling}
Let $\mu:= np$.  Then
\[ \hat{\mu}_i^{upper} = n \hat{p}_i^{upper} \overset{\mathcal{D}}{\rightarrow} Gamma(Y_{i} + 1, 1) ~~~~~~~ \hat{\mu}_i^{lower} =  n \hat{p}_i^{lower} \overset{\mathcal{D}}{\rightarrow} Gamma(Y_{i} , 1), \] 
as $n \rightarrow \infty$. 
\end{lemma}\begin{proof}
See, for instance, \cite{Gut2005}.
\end{proof}  



Our aim is to develop a method that can simulate draws from the GFD by calculating sets with elements of the form $\hat{n} \times (\hat{\mu}^{lower}, \hat{\mu}^{upper}) \in \nn \times \mathcal{B} (\rr_{+})$.    A simplistic approach to generating values from the GFD would be to repeatedly re-sample the set of random variables $\textbf{U}$ on the space $[0,1]^{m}$ and in turn recalculate the paired sets that satisfy \eqref{eq:BinNDGA}. There are two issues that arise for this approach.  First,  there is no upper bound on the set of possible $n$ values, and second, it is possible that a majority of the uniform sets $\mathbf{U}$ drawn completely randomly will yield no solution to \eqref{eq:BinNDGA}. To address these issues we develop a stopping criterion for searching the sample space in the $n$ direction and we develop a more computationally efficient way to randomly select uniform values $\mathbf{U}$.

We use the following protocol for choosing a range of potential $n$ values.  While there are many ways to choose a first candidate $\hat{n}$, we use the estimator suggested by \cite{DasGupta}.  From this first candidate we check each subsequent $n$ by increment of one and calculate $(\hat{\mu}^{lower}, \hat{\mu}^{upper})$.  We develop a stopping rule for $\hat{n}$ based on the asymptotic result in Lemma \ref{lem:stirling}.  Given a precision parameter $\epsilon_2 > 0$, we stop looking for further $\hat{n}$ values whenever
\[ |\hat{\mu}^{upper}_i - H^{-1}(1 - U_{i})_{Y_{i} + 1, 1}| < \epsilon_2 ~~~\text{and}~~~ |\hat{\mu}_i^{lower}  - H^{-1}(1 - U_{i})_{Y_{i}, 1}| < \epsilon_2,  \]
for all $i \in \{1, \dots, m\}$, where $H^{-1}_{\alpha, \beta}$ is the quantile function of a $Gamma(\alpha, \beta)$ random variable.  We stop searching for values of $n$ that work as soon as we are sufficiently close to the limiting Gamma distribution, which does not depend on $n$.  This stopping criterion leverages our assumption that our target data is binomial, and not Poisson, to reduce the space of possible $n$ values to a finite set.

We implement an Gibbs sampler, sampling from the uniform distribution $\textbf{U}$ randomly selected in such a way that there is a new solution set that is non-empty, i.e., satisfying  Equation \eqref{eq:DGEbeta}, or equivalently
\[ G_{n-Y_i,Y_i + 1} (1-p) \geq U_i > G_{n-Y_i+1,Y_i}(1-p),\quad i=1,\ldots,m. \]
In particular, each value $U_i$ is resampled iteratively by randomly selecting a new value from its conditional distribution given all the other $U$s; a uniform distribution on
\[ \big( \min\{ G_{n-Y_i+1,Y_i}\left(1-\frac{\tilde{\mu}^{lower}_{\tilde{n},j}}{\tilde{n}}\right): j \neq i\}, \max\{ G_{n-Y_i,Y_i+1}\left(1-\frac{\tilde{\mu}^{lower}_{\tilde{n},j}}{\tilde{n}}\right): j\neq i\}\big],\]
where $\tilde{n}$, $\tilde{\mu}^{lower}_{\tilde{n},j}$, and $\tilde{\mu}^{upper}_{\tilde{n},j}$ for $j = \{1,\ldots,m\}\backslash \{i\}$ are the solution set that we would get if we were to remove the $i$th observation from our data.
To speed up the computation we actually resample $U_i$s in batches, so that all $U_i$s corresponding to the same observed value of $Y_i$ are sampled together. 

Our investigations of the above method of resampling the uniform values $\mathbf{U}$ have revealed that there is a need for steps that would result large changes to  $n$ and $\mu$.   To allow for this, we add two Metropolis-Hastings steps at the end of each Gibbs sampler scan, one in the $n$ direction and one in the $\mu$ direction.  

The MH step in the $\mu$ direction was implemented using the following algorithm.  Assume that you already have finished the above protocol for choosing a new set of uniform values $\mathbf{U}$ and that has a corresponding solution set \[\bigcup_{\hat{n}} \{\hat{n}\} \times (\max\{\hat{\mu}_{1,\hat{n}}^{lower},\ldots,\hat{\mu}_{m,\hat{n}}^{lower}\},\min\{\hat{\mu}_{1,\hat{n}}^{upper},\ldots,\hat{\mu}_{m,\hat{n}}^{upper}\}].\]  Randomly select a single value $\mu^*$ from the set of $\hat{\mu}$ values associated with the smallest $\hat{n}$ \[(\max\{\hat{\mu}_{1,\min\{\hat{n}\}}^{lower},\ldots,\hat{\mu}_{m,\min\{\hat{n}\}}^{lower}\},\min\{\hat{\mu}_{1,\min\{\hat{n}\} }^{upper},\ldots,\hat{\mu}_{m,\min\{\hat{n}\}}^{upper}\}].\]  {\myfont Next, define a new value, $\mu^\star$, that is drawn from the proposal distribution $N(\mu^*, \sigma^2)$ for some predetermined $\sigma$.}  Calculate a set of uniform values, $\mathbf{U}^\star$, such that $\mu^\star$ is contained in the solution set.  Let $\dot{n}, \{\dot{\mu}^{lower}_{\min\{\dot{n}\},j}\}_{j=1}^m$, and $\{\dot{\mu}^{upper}_{\min\{\dot{n}\},j}\}_{j=1}^m$ be the values in the solution set created by using Equation \eqref{eq:DGEbeta} with $\mathbf{U}^\star$. The acceptance ratio for these new $\mathbf{U}^\star$ will then be 
\[ \frac{(\dot{\mu}^{lower}_{\min\{\dot{n}\},j}- \dot{\mu}^{upper}_{\min\{\dot{n}\},j}) \prod_{i=1}^m \left( G_{n-Y_i,Y_i+1}\left(1-\frac{\hat{\mu}^{lower}_{\hat{n},i}}{\hat{n}}\right) -  G_{n-Y_i+1,Y_i}\left(1-\frac{\hat{\mu}^{lower}_{\hat{n},i}}{\hat{n}}\right) \right)}{ (\hat{\mu}^{lower}_{\min\{\hat{n}\},j}- \hat{\mu}^{upper}_{\min\{\hat{n}\},j}) \prod_{i=1}^m \left( G_{n-Y_i,Y_i+1}\left(1-\frac{\dot{\mu}^{lower}_{\dot{n},i}}{\dot{n}}\right) -  G_{n-Y_i+1,Y_i}\left(1-\frac{\dot{\mu}^{lower}_{\dot{n},i}}{\dot{n}}\right) \right)}, \]

The MH step in the $n$ direction was implemented using the following algorithm.  Like the MH step in the $\mu$ direction, start by randomly selecting a single $\mu^*$ value from set of $\hat{\mu}$ values associated with the smallest $\hat{n}$, $\mu^* \in (\hat{\mu}^{lower}_{\min\{\hat{n}\},j}, \hat{\mu}^{upper}_{\min\{\hat{n}\},j})$.  Then, let $n^\star := \min\{\hat{n}\} + 1-2X$, where $X \sim Bernoulli(1/2)$, be your new proposed $n$ value.  Next, find a set of uniform values $\mathbf{U}^\star$ such that $n^\star$ and $\mu^\star$ are both in the solution set.  Let $\dot{n}, \{\dot{\mu}^{lower}_{\min\{\dot{n}\},j}\}_{j=1}^m$, and $\{\dot{\mu}^{upper}_{\min\{\dot{n}\},j}\}_{j=1}^m$ be the values in the solution set created by using Equation \eqref{eq:DGEbeta} with $\mathbf{U}^\star$.  The acceptance ratio of $\mathbf{U}^\star$ has the same form as the acceptance ratio for the MH step in the $\mu$ direction.

The specific details of the algorithm can be read from our implementation (in R) on GitHub [https://github.com/sirmurphalot/IntroductionGFI].  On this page we have also posted a full pseudocode version of the algorithm for the curious reader.
\end{appendix}

\bibliographystyle{rss}
\bibliography{examples}       

\begin{thebibliography}{44}
\expandafter\ifx\csname natexlab\endcsname\relax\def\natexlab#1{#1}\fi
\expandafter\ifx\csname url\endcsname\relax
  \def\url#1{\texttt{#1}}\fi
\expandafter\ifx\csname urlprefix\endcsname\relax\def\urlprefix{URL }\fi

\bibitem[{Ahrens and Pincus(1981)}]{ahrens1981}
Ahrens, H. and Pincus, R. (1981) On two measures of unbalancedness in a one-way
  model and their relation to efficiency.
\newblock \textit{Biometrical Journal}, \textbf{23}, 227--235.
\newblock
  \urlprefix\url{https://onlinelibrary.wiley.com/doi/abs/10.1002/bimj.4710230302}.

\bibitem[{Berger \textit{et~al.}(2020{\natexlab{a}})Berger, Sun and
  Song}]{BergerSunSong2018a}
Berger, J.~O., Sun, D. and Song, C. (2020{\natexlab{a}}) Bayesian analysis of
  the covariance matrix of a multivariate normal distribution with a new class
  of priors.
\newblock \textit{Annals of Statistics}, \textbf{48}, 2381--2403.

\bibitem[{Berger \textit{et~al.}(2020{\natexlab{b}})Berger, Sun and
  Song}]{BergerSunSong2018b}
--- (2020{\natexlab{b}}) An objective prior for hyperparameters in normal
  hierarchical models.
\newblock \textit{Journal of Multivariate Analysis}, 104606.

\bibitem[{Cisewski and Hannig(2008)}]{cisewski2012}
Cisewski, J. and Hannig, J. (2008) Generalized fiducial inference for normal
  linear mixed models.
\newblock \textit{The Annals of Statistics}, \textbf{40}, 2102--2127.

\bibitem[{DasGupta and Rubin(2004)}]{DasGupta}
DasGupta, A. and Rubin, H. (2004) Estimation of binomial parameters when both
  n,p are unknown.
\newblock \textit{Journal of Statistical Planning and Inference}, \textbf{130},
  391--404.

\bibitem[{Dempster(2008)}]{dempster}
Dempster, A. (2008) The dempster–shafer calculus for statisticians.
\newblock \textit{International Journal of Approximate Reasoning}, \textbf{48},
  365 -- 377.
\newblock
  \urlprefix\url{http://www.sciencedirect.com/science/article/pii/S0888613X07000278}.
\newblock In Memory of Philippe Smets (1938–2005).

\bibitem[{E \textit{et~al.}(2008)E, Hannig and Iyer}]{e2008}
E, L., Hannig, J. and Iyer, H. (2008) Fiducial intervals for variance
  components in an unbalanced two-component normal mixed linear model.
\newblock \textit{Journal of the American Statistical Association},
  \textbf{103}, 854--865.
\newblock \urlprefix\url{https://doi.org/10.1198/016214508000000229}.

\bibitem[{Edlefsen \textit{et~al.}(2009)Edlefsen, Liu and
  Dempster}]{edlefsen2009}
Edlefsen, P., Liu, C. and Dempster, A. (2009) Estimating limits from poisson
  counting data using dempster-shafer analysis.
\newblock \textit{Annals of Applied Statistics}, \textbf{3}, 764--790.

\bibitem[{Eves(1996)}]{howard1996}
Eves, H. (1996) \textit{Elementary Matrix Theory}.
\newblock Allyn and Bacon, Inc., 1 edn.
\newblock 265-267.

\bibitem[{{Fazel} \textit{et~al.}(2003){Fazel}, {Hindi} and {Boyd}}]{Fazel2003}
{Fazel}, M., {Hindi}, H. and {Boyd}, S.~P. (2003) Log-det heuristic for matrix
  rank minimization with applications to hankel and euclidean distance
  matrices.
\newblock In \textit{Proceedings of the 2003 American Control Conference,
  2003.}, vol.~3, 2156--2162 vol.3.

\bibitem[{Fisher(1935)}]{fisher1935}
Fisher, R. (1935) The fiducial argument...
\newblock \textit{Annals of Eugenics}, \textbf{6}, 391--398.

\bibitem[{Grafarend \textit{et~al.}(2003)Grafarend, Krumm and
  Schwarze}]{Forstner2003}
Grafarend, E.~W., Krumm, F.~W. and Schwarze, V.~S. (eds.) (2003) \textit{A
  Metric for Covariance Matrices}, 299--309.
\newblock Berlin, Heidelberg: Springer Berlin Heidelberg.

\bibitem[{Gut(2005)}]{Gut2005}
Gut, A. (2005) \textit{Probability: A Graduate Course}.
\newblock Springer.

\bibitem[{Hannig(2009)}]{hannig2009AP}
Hannig, J. (2009) {On Generalized Fiducial Inference}.
\newblock \textit{Statistica Sinica}, \textbf{19}, 491--544.

\bibitem[{Hannig \textit{et~al.}(2003)Hannig, E, Abdel-Karmin and
  Iyer}]{hannig2006a}
Hannig, J., E, L., Abdel-Karmin, A. and Iyer, H. (2003) Simultaneous fiducial
  generalized confidence intervals for ratios of means of lognormal
  distributions.
\newblock \textit{Austrian Journal of Statistics}, \textbf{35}, 261--269.

\bibitem[{Hannig \textit{et~al.}(2007)Hannig, E, Abdel-Karmin and
  Iyer}]{hannig2007}
--- (2007) Fiducial approach to uncertainty asessment: Accounting for error due
  to instrument resolution.
\newblock \textit{Metrologia}, \textbf{44}, 476--483.

\bibitem[{Hannig \textit{et~al.}(2016)Hannig, Iyer, Lai and Lee}]{GenFid}
Hannig, J., Iyer, H., Lai, R. C.~S. and Lee, T. C.~M. (2016) Generalized
  fiducial inference: A review and new results.
\newblock \textit{Journal of the American Statistical Association},
  \textbf{111}, 1346--1361.
\newblock \urlprefix\url{https://doi.org/10.1080/01621459.2016.1165102}.

\bibitem[{Hannig and Lee(2009)}]{hannig2009}
Hannig, J. and Lee, T. (2009) Generalized fiducial inference for wavelet
  regression.
\newblock \textit{Biometrika}, \textbf{96}, 847--860.

\bibitem[{Hannig \textit{et~al.}(2013)Hannig, Wang and Iyer}]{hannig2003}
Hannig, J., Wang, C.~M. and Iyer, H.~K. (2013) Uncertainty calculation for the
  ratio of dependent measurements.
\newblock \textit{Metrologia}, \textbf{4}, 177--186.

\bibitem[{Henderson and Searle(1979)}]{henderson1979}
Henderson, H.~V. and Searle, S.~R. (1979) Vec and vech operators for matrices,
  with some uses in jacobians and multivariate statistics.
\newblock \textit{The Canadian Journal of Statistics / La Revue Canadienne de
  Statistique}, \textbf{7}, 65--81.

\bibitem[{Horn and Johnson(2012)}]{horn2012}
Horn, R.~A. and Johnson, C.~R. (2012) \textit{Matrix Analysis}.
\newblock USA: Cambridge University Press, 2nd edn.

\bibitem[{Iyer \textit{et~al.}(2004)Iyer, Wang and Mathew}]{iyer2004}
Iyer, H., Wang, J.~C. and Mathew, T. (2004) Models and confidence intervals for
  true values in interlaboratory trials.
\newblock \textit{Journal of the American Statistical Association},
  \textbf{99}, 1060--1071.

\bibitem[{Konno(1995)}]{Konno1995}
Konno, Y. (1995) Estimation of a normal covariance matrix with incomplete data
  under stein's loss.
\newblock \textit{Journal of Multivariate Analysis}, \textbf{52}, 308 -- 324.
\newblock
  \urlprefix\url{http://www.sciencedirect.com/science/article/pii/S0047259X85710160}.

\bibitem[{Liu and Hannig(2016)}]{liu2016}
Liu, Y. and Hannig, J. (2016) Generalized fiducial inference for binary
  logistic item response models.
\newblock \textit{Psychometrica}, \textbf{81}, 290--324.

\bibitem[{McNally \textit{et~al.}(2003)McNally, Iyer and Mathew}]{mcnally2003}
McNally, R., Iyer, H. and Mathew, T. (2003) Tests for individual and population
  bioequivalence based on generalized p-values.
\newblock \textit{Statistics in Medicine}, \textbf{22}, 31--53.

\bibitem[{Neupert \textit{et~al.}(2020)Neupert, Growney, Zhu, Sorensen, Smith
  and Hannig}]{NeupertEtAl2020}
Neupert, S.~D., Growney, C.~M., Zhu, X., Sorensen, J.~K., Smith, E.~L. and
  Hannig, J. (2020) Bff: Bayesian, fiducial, and frequentist analysis of
  cognitive engagement among cognitively impaired older adults.
\newblock Submitted for publication.

\bibitem[{O'Dorney(2014)}]{odorney2014}
O'Dorney, E. (2014) Minimizing the cayley transform of an orthogonal matrix by
  multiplying by signature matrices.
\newblock \textit{Linear Algebra and its Applications}, \textbf{448}, 97 --
  103.
\newblock
  \urlprefix\url{http://www.sciencedirect.com/science/article/pii/S0024379514000615}.

\bibitem[{Perron(1992)}]{perron1993}
Perron, F. (1992) Minimax estimators of a covariance matrix.
\newblock \textit{Journal of Multivariate Analysis}, \textbf{43}, 16--28.
\newblock
  \urlprefix\url{https://www.sciencedirect.com/science/article/pii/0047259X9290108R}.

\bibitem[{Schweder and Hjort(2016)}]{SchwederHjortBook}
Schweder, T. and Hjort, N.~L. (2016) \textit{Confidence, likelihood,
  probability}, vol.~41.
\newblock Cambridge University Press.

\bibitem[{Shafer(1976)}]{Shafer1976}
Shafer, G. (1976) \textit{{A mathematical theory of evidence}}.
\newblock Princeton, New Jersey: Princeton University Press.

\bibitem[{Shi \textit{et~al.}(2021)Shi, Hannig, Lai and
  Lee}]{ShiHannigLaiLee2017}
Shi, W.~J., Hannig, J., Lai, R. C.~S. and Lee, T. C.~M. (2021) Covariance
  estimation via fiducial inference.
\newblock \textit{Statistical Theory and Related Fields}, \textbf{5}, 316--331.
\newblock \urlprefix\url{https://doi.org/10.1080/24754269.2021.1877950}.

\bibitem[{{Stan Development Team}(2020)}]{STAN}
{Stan Development Team} (2020) {RStan}: the {R} interface to {Stan}.
\newblock \urlprefix\url{http://mc-stan.org/}.
\newblock R package version 2.21.2.

\bibitem[{Wadsworth(1960)}]{wadsworth}
Wadsworth, G.~P. (1960) \textit{Introduction to probability and random
  variables}.
\newblock New York: McGraw-Hill, 1 edn.
\newblock P. 52.

\bibitem[{Wandler and Hannig(2012{\natexlab{a}})}]{wandler2012b}
Wandler, D. and Hannig, J. (2012{\natexlab{a}}) Generalized fiducial confidence
  intervals for extremes.
\newblock \textit{Extremes}, \textbf{15}, 67--87.

\bibitem[{Wandler and Hannig(2006)}]{glagovskiy2006}
Wandler, D.~V. and Hannig, J. (2006) Construction of fiducial confidence
  intervals for the mixture of cauchy and normal distributions.
\newblock \textit{Master's Thesis, Department of Statistics, Colorado State
  University}.

\bibitem[{Wandler and Hannig(2011)}]{wandler2011}
--- (2011) Fiducial inference on the maximum mean of a multivariate normal
  distribution.
\newblock \textit{Journal of Multivariate Analysis}, \textbf{102}, 87--104.

\bibitem[{Wandler and Hannig(2012{\natexlab{b}})}]{wandler2012a}
--- (2012{\natexlab{b}}) A fiducial approach to multiple comparisons.
\newblock \textit{Journal of Statistical Planning and Inference}, \textbf{142},
  878--895.

\bibitem[{Wang and Iyer(2006{\natexlab{a}})}]{wang2006a}
Wang, C.~M. and Iyer, H.~K. (2006{\natexlab{a}}) Propagation of uncertainties
  in measurements using generalized inference.
\newblock \textit{Metrologia}, \textbf{42}, 145--153.

\bibitem[{Wang and Iyer(2006{\natexlab{b}})}]{wang2006b}
--- (2006{\natexlab{b}}) Uncertainty analysis for vector measurands using
  fiducial inference.
\newblock \textit{Metrologia}, \textbf{43}, 486--494.

\bibitem[{Wang \textit{et~al.}(2012)Wang, Hannig and Iyer}]{wang2012b}
Wang, J., Hannig, J. and Iyer, H. (2012) Pivotal methods in the propagation of
  distributions.
\newblock \textit{Metrologia}, \textbf{49}, 382--389.

\bibitem[{Williams and Hannig(2019)}]{williams2019}
Williams, J.~P. and Hannig, J. (2019) Nonpenalized variable selection in
  high-dimensional linear model settings via generalized fiducial inference.
\newblock \textit{The Annals of Statistics}, \textbf{47.3}, 1723--1753.

\bibitem[{Williams \textit{et~al.}(2019)Williams, Xie and
  Hannig}]{williams2019b}
Williams, J.~P., Xie, Y. and Hannig, J. (2019) The eas approach for graphical
  selection consistency in vector autoregression models.
\newblock \textit{arXiv preprint arXiv:1906.04812}.

\bibitem[{Yager and Liu(2010)}]{Yager2010}
Yager, R.~R. and Liu, L. (2010) \textit{Classic Works of the Dempster-Shafer
  Theory of Belief Functions}.
\newblock Springer Publishing Company, Incorporated, 1st edn.

\bibitem[{Yang and Berger(1994)}]{YangBerger1994}
Yang, R. and Berger, J.~O. (1994) {Estimation of a Covariance Matrix Using the
  Reference Prior}.
\newblock \textit{The Annals of Statistics}, \textbf{22}, 1195--1211.

\end{thebibliography}


\end{document}